\documentclass[]{article}
\usepackage[left=1in,right=1in, top=1in,bottom=1in]{geometry}     
\usepackage[dvipsnames]{xcolor}
\usepackage{xcolor}
\usepackage{cite}
\usepackage{color}
\usepackage[T1]{fontenc}
\usepackage[latin9]{inputenc}
\usepackage{amsmath}
\usepackage{amssymb}
 
\usepackage{amsthm} 
\usepackage{mathtools,cuted}
\usepackage{mathrsfs}
\usepackage{bbm}

\allowdisplaybreaks
\usepackage{lipsum}
\usepackage{multicol}
\usepackage{changepage}

\newtheorem{thm}{\bf Theorem}
\newtheorem{lem}{\bf Lemma}

\theoremstyle{definition}

\usepackage{dsfont}
\usepackage{relsize}
\usepackage{subdepth}

\usepackage{url}

\DeclareFontFamily{OT1}{pzc}{}
\DeclareFontShape{OT1}{pzc}{m}{it}{<-> s * [1.200] pzcmi7t}{}
\DeclareMathAlphabet{\mathpzc}{OT1}{pzc}{m}{it}

\usepackage{filecontents}

\usepackage{stfloats}

\renewcommand\footnotemark{}

\usepackage{bigints}

\usepackage{scalerel}

\usepackage[framemethod=tikz]{mdframed}

\newmdenv[
  hidealllines=true,
  backgroundcolor=blue!10,
  innerleftmargin=8pt,
  innerrightmargin=8pt,
  innertopmargin=0pt,
  innerbottommargin=6pt,
  leftmargin=-0pt,
  rightmargin=-0pt
]{shadedbox}

\definecolor{airforceblue}{rgb}{0.36, 0.54, 0.66}
\definecolor{ballblue}{rgb}{0.13, 0.67, 0.8}

\definecolor{alizarin}{rgb}{0.82, 0.1, 0.26}
\definecolor{asparagus}{rgb}{0.53, 0.66, 0.42}
\definecolor{applegreen}{rgb}{0.55, 0.71, 0.0}
\definecolor{armygreen}{rgb}{0.29, 0.33, 0.13}
\definecolor{amber(sae/ece)}{rgb}{1.0, 0.49, 0.0}
\definecolor{coquelicot}{rgb}{1.0, 0.22, 0.0}
\definecolor{ao(english)}{rgb}{0.0, 0.5, 0.0}
\definecolor{amber}{rgb}{1, 0.6, 0}

\makeatletter
\let\NAT@parse\undefined
\makeatother
\usepackage{hyperref}
\hypersetup{
    colorlinks=true,
    linkcolor=red,
    filecolor=magenta,      
    urlcolor=cyan,
    citecolor=ballblue
}

\usepackage{anyfontsize}
\usepackage{cleveref}

\newcommand{\sev}{\mathrm{sev}}
\newcommand{\mse}{\mathrm{mse}}

\newcommand{\cc}{\boldsymbol{\Sigma}_{\boldsymbol{X}|\boldsymbol{Y}}}

\newcommand{\cm}{\mathbb{E}\{ \boldsymbol{X}|\boldsymbol{Y}\}}
\newcommand{\cmtwo}{\mathbb{E}\{ \boldsymbol{||X||}_2^2|\boldsymbol{Y}\}}
\newcommand{\cmthree}{\mathbb{E}\{ \boldsymbol{||X||}_2^2\boldsymbol{X}|\boldsymbol{Y}\}}
\newcommand{\xemi}{\hat{\boldsymbol{X}}_{\mu}^{*}}
\newcommand{\xemit}{\hat{\boldsymbol{X}}_{{\mu}^{ \prime}}^{*}}
\newcommand{\xeo}{\hat{\boldsymbol{X}}_{0}^{*}}
\newcommand{\xeinf}{\hat{\boldsymbol{X}}_\infty^{*}}

\newcommand{\dx}{\widehat{\boldsymbol{\Delta}\boldsymbol{X}}}

\newcommand{\sminy}{\sigma_{\min}{(\boldsymbol{Y})}}
\newcommand{\smaxy}{\sigma_{\max}{(\boldsymbol{Y})}}
\newcommand{\siy}{\sigma_{i}{(\boldsymbol{Y})}}

\title{
Uncertainty Principles in Risk-Aware Statistical Estimation}

\author{
Nikolas P. Koumpis and Dionysios S. Kalogerias
\thanks{The Authors are with the Department of EE, Yale University,
New Haven, CT. email: \{dionysis.kalogerias, nikolaos.koumpis\}@yale.edu.}}

\begin{document}
\date{}

\maketitle
\thispagestyle{empty}

\begin{abstract}
We present a new uncertainty principle for risk-aware statistical estimation, effectively quantifying the inherent trade-off between mean squared error ($\mse$) and risk, the latter measured by the associated average predictive squared error variance ($\sev$), for every admissible estimator of choice. Our uncertainty principle has a familiar form and resembles fundamental and classical results arising in several other areas, such as the Heisenberg principle in statistical and quantum mechanics, and the Gabor limit (time-scale trade-offs) in harmonic analysis. In particular, we prove that, provided a joint generative model of states and observables, the product between $\mse$ and $\sev$ is bounded from below by a computable model-dependent constant, which is explicitly related to the Pareto frontier of a recently studied $\sev$-constrained minimum $\mse$ (MMSE) estimation problem. Further, we show that the aforementioned constant is inherently connected to 
 an intuitive new and rigorously topologically grounded statistical measure of distribution skewness in multiple dimensions, consistent with Pearson's moment coefficient of skewness for variables on the line. Our results are also illustrated via numerical simulations.
 
\end{abstract}

\section{Introduction}\label{intro}
Designing decision rules aiming for least expected losses is a standard and commonly employed objective in statistical learning, estimation, and control. Still, achieving optimal performance on average is insufficient without safeguarding against less probable though statistically significant, i.e., \textit{risky}, events, and this is especially pronounced in modern, critical applications. Examples appear naturally in many areas, including robotics \cite{kim2019bi}, \cite{wang2020fast}, wireless communications and networking \cite{ma2018risk},\cite{bennis2018ultrareliable}, edge computing \cite{Li2020}, health \cite{Cardoso2019}, and finance \cite{dentcheva2017statistical}, to name a few.
Indeed, 
risk-neutral decision policies smoothen unexpected events by construction, thus exhibiting potentially large
statistical performance volatility, since the latter remains uncontrolled. 
In such situations, risk-aware decision rules are highly desirable as they systematically guarantee robustness, in the form of various operational specifications, such as safety \cite{Chapman2019,Samuelson2018}, fairness \cite{williamson2019fairness, tao2020incorporating}, distributional
robustness \cite{Gurbuzbalaban2020, curi2019adaptive}, and prediction error stability \cite{kalogerias2020noisy}.
\par
In the realm of Bayesian mean squared error ($\mathrm{mse}$) statistical estimation, 
a risk-constrained reformulation of the standard  minimum $\mathrm{mse}$ (MMSE) estimation problem was recently proposed in \cite{kalogerias2020better}, where, given a generative model (i.e., distribution) of states and observables, risk is measured by the average predictive squared error variance ($\mathrm{sev}$) associated with every feasible square integrable (i.e., admissible) estimator. Quite remarkably, such a constrained functional estimation problem admits a unique closed-form solution;
as compared with classical risk-neutral MMSE estimation (i.e., conditional mean), the optimal risk-aware estimator nonlinearly interpolates between the risk-neutral MMSE estimator (i.e., conditional mean) and a new, maximally risk-aware statistical estimator, minimizing average errors while constraining risk under a designer-specified threshold.

From the analysis presented in \cite{kalogerias2020better}, it becomes evident that low-risk estimators deteriorate performance on average and vice-versa. However, although $\mathrm{mse}$ and $\mathrm{sev}$ (i.e., risk) are shown to trade between each other within the class of optimal risk-aware estimators proposed in  \cite{kalogerias2020better}, a mathematical statement that expresses this fundamental interplay for \textit{general} estimators is non-trivial and currently unknown.
This paper is precisely on the discovery, quantification and analysis of this interplay. Our contributions are as follows.

\textbf{--A New $\mathrm{mse}$/$\mathrm{sev}$ Uncertainty Principle (Section \ref{section3}).}
We quantify the trade-off between $\mathrm{mse}$ and $\mathrm{sev}$ associated with any square integrable estimator of choice by bounding their product by a model-dependent, estimator-independent characteristic constant. This fundamental lower bound, which we call the \textit{optimal trade-off}, is always attained within the class of optimal risk-aware estimators of \cite{kalogerias2020better}, and provides a universal benchmark of the trade-off efficiency of every possible admissible estimator. Our uncertainty relation comes in the natural form of an \textit{uncertainty principle};
similar relations are met in different contexts, e.g., in statistical  mechanics (Heisenberg principle) \cite{cohen2008quantum} and harmonic analysis (Gabor limit) \cite{vakmann2012sophisticated}. 
In essence, uncertainty principles are bounds on the concentration or spread of a quantity in two different domains. In our case, $\mathrm{sev}$ measures the squared error statistical spread, while $\mathrm{mse}$ measures the squared error average (expected) value. 
Our uncertainty principle states that, \textit{in general}, both quantities cannot be simultaneously small, let alone minimized; in the latter case exceptions exist, herein called the class of \textit{skew-symmetric} models. In fact, conditionally Gaussian models are a canonical example in this exceptional class. 



\textbf{--Hedgeable Risk Margins and Lower Bound Characterization (Sections \ref{section4}-\ref{section5}).}
We present an intuitive geometric interpretation of the class of risk-aware estimators of \cite{kalogerias2020better}, inherently related to the optimal trade-off involved in our uncertainty principle. We define a new quantity, called the \textit{expected hedgeable risk margin} associated with the underlying generative model by projecting the stochastic parameterized curve induced by the class of risk-aware estimators of \cite{kalogerias2020better} onto the line that links the risk-neutral (i.e., MMSE) with the maximally risk-averse estimator.
Intuitively, such a projection expresses the margin to potentially counteract against risk (as measured by the $\mathrm{sev}$), on average relative to the distribution of the observables.
Subsequently, we show that, under mild assumptions, the optimal trade-off is \textit{order- equivalent} to the corresponding expected risk margin. We do this by proving explicit and order-matching upper and lower bounds on the optimal trade-off that depend strictly proportionally to the expected risk margin. The importance of this result is that a large (small) risk margin implies a large (small) optimal trade-off, and vice versa.

\textbf{--Topological/Statistical Interpretation of Risk Margins, and Skewness in High Dimensions (Section \ref{section6}).} The significance of the risk-margin functional is established by showing that it admits a dual topological and statistical interpretation, within a rigorous technical framework.
First, we prove that the space of all generative models with a finite risk-margin becomes a topological space endowed with a (pseudo)metric, the latter induced by a certain risk-margin-related functional. This functional vanishes for all skew-symmetric models, and therefore is rigorously interpretable as a distance to all members of this exceptional class (via the (pseudo)metric).
Simultaneously, the  aforementioned risk-margin-related functional corresponds to an intuitive model statistic which can be regarded as a generalized measure of distribution skewness, 
consistent with the familiar Pearson's moment coefficient of skewness, to which it reduces exactly for totally unobservable variables on the line. 
Similarly, the induced (pseudo)metric may be
 regarded as a measure of the \textit{relative skewness}
between (filtered) distributions.









Lastly, our results are supported by indicative numerical examples, along with a relevant discussion (Section \ref{section7}).

\section{SEV-Constrained MMSE Estimation}\label{section2}

On a probability space $(\Omega,\mathscr{F},{\mathpzc P})$, consider random elements
$\boldsymbol{X}:\Omega\rightarrow\mathbb{R}^{n}$ and $\boldsymbol{Y}:\Omega\rightarrow\mathbb{R}^{m}$ following a joint Borel probability measure ${\cal P}_{(\boldsymbol{X},\boldsymbol{Y})}\equiv\mathcal{P}$. Intuitively, $\boldsymbol{X}$ may be thought of as a \textit{hidden} random state of nature, and $\boldsymbol{Y}$ as the corresponding \textit{observables}. Also, hereafter, let ${\cal L}_{2|\mathscr{Y}}$ be the space of square-integrable
$\mathscr{Y}\triangleq\sigma\{\boldsymbol{Y}\}\text{-measurable}$ estimators (i.e., deterministic functions of the observables). 
Provided a generative model ${\cal P}_{(\boldsymbol{X},\boldsymbol{Y})}$, we consider the \textit{mean squared error} and \textit{squared
error variance} functionals $\mathrm{mse}:{\cal L}_{2|\mathscr{Y}}\rightarrow\mathbb{R}_{+}$
and $\mathrm{sev}:{\cal L}_{2|\mathscr{Y}}\rightarrow\mathbb{R}_{+}$
defined respectively as
\begin{flalign}\label{eq12}
\mathrm{mse}(\hat{\boldsymbol{X}}) & \hspace{-1bp}\triangleq\hspace{-1bp}\mathbb{E}\{\Vert\boldsymbol{X}-\hat{\boldsymbol{X}}\Vert_{2}^{2}\},\quad\text{and}\\
\mathrm{sev}(\hat{\boldsymbol{X}}) & \hspace{-1bp}\triangleq\hspace{-1bp}\mathbb{E}\{\mathbb{V}_{\boldsymbol{Y}}\{\Vert\boldsymbol{X}-\hat{\boldsymbol{X}}\Vert_{2}^{2}\}\hspace{-1pt}\} \\
 & \hspace{-1bp}\equiv\hspace{-1bp}\mathbb{E}\big\{\hspace{0pt}\hspace{0pt}\mathbb{E}\big\{\hspace{-1bp}\big(\hspace{-0.5pt}\Vert\boldsymbol{X}-\hat{\boldsymbol{X}}\Vert_{2}^{2}\hspace{-1bp}-\hspace{-1bp}\mathbb{E}\{\Vert\boldsymbol{X}-\hat{\boldsymbol{X}}\Vert_{2}^{2}|\boldsymbol{Y}\}\big)^{2}\big| \boldsymbol{Y}\big\}\hspace{-1.5bp}\big\},\nonumber
\end{flalign}
where $\hat{\boldsymbol{X}}\in{\cal L}_{2|\mathscr{Y}}$. Note that both functionals $\mathrm{mse}$ and $\mathrm{sev}$ are \textit{law invariant}, i.e., they depend exclusively on ${\cal P}_{(\boldsymbol{X},\boldsymbol{Y})}$ \cite{ShapiroLectures_2ND}. As such, they
may be equivalently thought of as mappings whose domain is the space of Borel probability measures on the product space $\mathbb{R}^n \times \mathbb{R}^m$.

While $\mathrm{mse}$ quantifies the squared error incurred by a given estimator $\hat{\boldsymbol{X}}$
on average and is a \textit{gold-standard} performance criterion in estimation and control \cite{Speyer2008STOCHASTIC}, $\mathrm{sev}$ quantifies the \textit{risk} of
$\hat{\boldsymbol{X}}$, as measured by the average predictive variance
of the associated instantaneous estimation error around its MMSE-sense
prediction given the observable $\boldsymbol{Y}$. In other words, $\mathrm{sev}$ quantifies the \textit{statistical variability} of $||\boldsymbol{X}-\hat{\boldsymbol{X}}||^2_2$ against the \textit{predictable statistical benchmark} $\mathbb{E}\{ ||\boldsymbol{X}-\hat{\boldsymbol{X}}||^2_2 | \boldsymbol{Y}\}$. Such statistical variability is left uncontrolled in standard MMSE estimation; in fact, this is a natural flaw of MMSE estimators (i.e., conditional means) by construction, resulting in statistically unstable prediction errors, especially in problems involving skewed and/or heavy-tailed distributions \cite{kalogerias2020better}.
%

To counteract risk-neutrality of MMSE estimators, a \textit{constrained} reformulation of the MMSE problem was recently introduced in \cite{kalogerias2020better}, where the $\mse$ is minimized subject to an explicit constraint on the associated $\sev$. The resulting risk-aware stochastic variational (i.e., functional) problem is
\begin{equation}
\begin{array}{rl}
\underset{\hat{\boldsymbol{X}}\in{\cal L}_{2|\mathscr{Y}}}{\mathrm{minimize}} & \mathrm{mse}(\hat{\boldsymbol{X}})\\
\mathrm{subject\,to} & \mathrm{sev}(\hat{\boldsymbol{X}})\le\varepsilon
\end{array},\label{eq:Base_Problem-V-3}
\end{equation}
where $\varepsilon>0$ is a user-prescribed tolerance. As problem \eqref{eq:Base_Problem-V-3} may be shown to be convex \cite{kalogerias2020better,Kalogerias2019b}, prominent role in the analysis of \eqref{eq:Base_Problem-V-3} plays its variational Lagrangian relaxation
\begin{equation}
\inf_{\hat{\boldsymbol{X}}\in{\cal L}_{2|\mathscr{Y}}}\mathrm{mse}(\hat{\boldsymbol{X}})+\mu\,\mathrm{sev}(\hat{\boldsymbol{X}}),\label{eq:Lagrangian}
\end{equation}
for fixed $\mu\ge0$, dependent of the particular $\varepsilon$ of choice.
By defining the third-order posterior statistic
\begin{equation}
    \boldsymbol{R}(\boldsymbol{Y})\triangleq\cmthree-\cmtwo\cm,\label{adonis}
\end{equation}
and under the mild condition that $\mathbb{E}\{\Vert \boldsymbol{X} \Vert_2^3|\boldsymbol{Y}\} \in {\cal L}_{2|\mathscr{Y}}$ (also assumed hereafter), an essentially unique optimal solution to (\ref{eq:Lagrangian}) may be expressed in closed
form as
\begin{align}
\hat{\boldsymbol{X}}_{\mu}^{*}(\boldsymbol{Y})=
\dfrac{\mathbb{E}\big\{\hspace{-1pt}\boldsymbol{X}|\boldsymbol{Y}\big\}+\mu\boldsymbol{R}(\boldsymbol{Y})}{\boldsymbol{I}+2\mu\boldsymbol{\Sigma}_{\boldsymbol{X}|\boldsymbol{Y}}},\label{eq:RA_MMSE}
\end{align}
for all $\mu\ge0$, where $\boldsymbol{\Sigma}_{\boldsymbol{X}|\boldsymbol{Y}}\succeq0$ denotes the conditional convariance of $\boldsymbol{X}$ relative to $\boldsymbol{Y}$.
When $\mu\equiv\infty$, we also \textit{define }the \textit{maximally risk-averse estimator} (corresponding to the tightest choice of $\varepsilon$)
\begin{equation}
\hspace{-3bp}\hat{\boldsymbol{X}}_{\infty}^{*}(\boldsymbol{Y})\triangleq\dfrac{1}{2}\boldsymbol{\Sigma}_{\boldsymbol{X}|\boldsymbol{Y}}^{\dagger}\boldsymbol{R}(\boldsymbol{Y})+
\boldsymbol{U}\begin{bmatrix}{\bf 0}_{r}\\
\big[\boldsymbol{U}^{\top}\mathbb{E}\big\{\hspace{-1pt}\boldsymbol{X}|\boldsymbol{Y}\big\}\big]_{n}^{r+1}\label{eq6}
\end{bmatrix},
\end{equation}
where $\boldsymbol{\Sigma}_{\boldsymbol{X}|\boldsymbol{Y}}^{\dagger}\succeq0$
denotes the Moore--Penrose pseudoinverse of $\boldsymbol{\Sigma}_{\boldsymbol{X}|\boldsymbol{Y}}$,
the latter with spectral decomposition $\boldsymbol{\Sigma}_{\boldsymbol{X}|\boldsymbol{Y}}\equiv\boldsymbol{U}\boldsymbol{\Lambda}\boldsymbol{U}^{\top}$
and of rank $r$.
It is then standard procedure to show that
$
\lim_{\mu\rightarrow\infty}\hat{\boldsymbol{X}}_{\mu}^{*}=\hat{\boldsymbol{X}}_{\infty}^{*}
$, implying that the paratererization ${\boldsymbol{X}}_{(\cdot)}^{*}$ is continuous on $[0,\infty]$.
Lastly, as also proved in \cite{kalogerias2020better}, whenever  $\mathcal{P}_{\boldsymbol{X}|\boldsymbol{Y}}$ satisfies the condition
\begin{equation}
\mathbb{E}\{\left(X_{i}-\mathbb{E}\left\{X_{i} \mid \boldsymbol{Y}\right\}\right)^{2}(\boldsymbol{X}-\mathbb{E}\{\boldsymbol{X} \mid \boldsymbol{Y}\}) \mid \boldsymbol{Y}\} \equiv \mathbf{0} \label{sufcond}
\end{equation}
for all $i \in \mathbb{N}_{n}^{+}$, it follows that
\begin{align}
    \xemi=\mathbb{E}\{ \boldsymbol{X}|\boldsymbol{Y}\},\quad\forall\mu \in [0,\infty].\label{SkS}
\end{align}
In particular, this is the case when $\mathcal{P}_{\boldsymbol{X}|\boldsymbol{Y}}$ is jointly Gaussian.
Hereafter, every generative model ${\cal P}_{(\boldsymbol{X},\boldsymbol{Y})}$ satisfying \eqref{SkS} for almost all $\boldsymbol{Y}$ will be called \textit{skew-symmetric}; this terminology is justified later in Section \ref{section5}.

\section{Uncertainty Principles}\label{section3}
Already from \eqref{eq:RA_MMSE} we can see that there is an inherent trade-off between $\mse$ and $\sev$ for the family of optimal estimators $\{\xemi\}_{\mu}$. Of course, the resulting $\mse$ and $\sev$ define the Pareto frontier of problem \eqref{eq:Lagrangian}. In this section, we quantify the $\mse$/$\sev$ trade-off for all admissible estimators. We do that by deriving a non-trivial lower bound on the product between $\mathrm{mse}$ and $\mathrm{sev}$. 
\par
We start by stating two technical lemmata, useful in our development. To this end, let $\smaxy$ denote the maximum eigenvalue of $\boldsymbol{\Sigma}_{\boldsymbol{X}|\boldsymbol{Y}}$, and define $\widehat{\boldsymbol{\Delta X}}\triangleq \xeo-\xeinf$.
\begin{lem}[\textbf{Monotonicity}]\label{lemma2}
The functions $\mathrm{mse}(\hat{\boldsymbol{X}}_{(\cdot)}^{*})$
and $\mathrm{sev}(\hat{\boldsymbol{X}}_{(\cdot)}^{*})$
are increasing and decreasing on $[0,\infty]$, respectively.
\end{lem}
\begin{lem} [\textbf{Continuity}]
The same functions $\mathrm{mse}(\hat{\boldsymbol{X}}_{(\cdot)}^{*})$
and $\mathrm{sev}(\hat{\boldsymbol{X}}_{(\cdot)}^{*})$ are continuous
on $[0,\infty]$ and Lipschitz continuous on $[0,\infty)$ with respective constants
\begin{equation}
\begin{aligned}
\mathcal{K}_{\mathrm{mse}}&=4\mathbb{E}\big\{\smaxy\big\|\widehat{\boldsymbol{\Delta X}} \big\|^2_2\big\} \quad\text{and}
\\
\mathcal{K}_{\mathrm{sev}}&=4\mathbb{E}\big\{\smaxy^2\big\|\widehat{\boldsymbol{\Delta X}}
\big\|^2_2\big\}.
\end{aligned}
\end{equation}
\end{lem}
Utilizing the lemmata above, we may now introduce the main result of the paper, which provides a new and useful characterization of the region of allowable $\mathrm{mse}$-$\mathrm{sev}$ combinations \textit{ever} possibly achievable by \textit{any} square-integrable estimator, given a generative model. Essentially, our result, which follows, quantifies that inherent trade-off between average estimation performance and risk.
\begin{thm}[\textbf{Uncertainty Principles}]\label{UP}
 Every admissible estimator $\hat{\boldsymbol{X}}\equiv\hat{\boldsymbol{X}}(\boldsymbol{Y})\in{\cal L}_{2|\mathscr{Y}}$
satisfies the lower bounds
\begin{equation}
{\mathrm{mse}(\hat{\boldsymbol{X}})\mathrm{sev}(\hat{\boldsymbol{X}})\ge\mathfrak{h}\ge\mathrm{mse}(\hat{\boldsymbol{X}}_{0}^{*})\mathrm{sev}(\hat{\boldsymbol{X}}_{\infty}^{*}),}
\end{equation}
where the characteristic number $\mathfrak{h}$ is given by
\begin{equation}
{\mathfrak{h}({\cal P})\equiv\mathrm{mse}(\hat{\boldsymbol{X}}_{\mu^{\star}}^{*})\mathrm{sev}(\hat{\boldsymbol{X}}_{\mu^{\star}}^{*}),}
\end{equation}
for any $\mu^{\star}\in\mathrm{arg}\mathrm{min}_{\mu\in[0,\infty]}\big\{\mathrm{mse}(\hat{\boldsymbol{X}}_{\mu}^{*})\mathrm{sev}(\hat{\boldsymbol{X}}_{\mu}^{*})\hspace{-1pt}\big\}\neq\emptyset$.
\end{thm}

\begin{proof}[Proof of Theorem 1]
We may examine the following three mutually exclusive cases:

\textit{Case 1:} $\mathrm{sev}(\hat{\boldsymbol{X}})\in(\mathrm{sev}(\hat{\boldsymbol{X}}_{\infty}^{*}),\mathrm{sev}(\hat{\boldsymbol{X}}_{0}^{*})]$.
Then, from the intermediate value theorem, it follows that there is $\mu_{\hat{\boldsymbol{X}}}\in[0,\infty)$
such that $\hat{\boldsymbol{X}}_{\mu_{\hat{\boldsymbol{X}}}}^{*}$
matches the performance of $\hat{\boldsymbol{X}}$, i.e., 
\begin{equation}
\mathrm{sev}(\hat{\boldsymbol{X}}_{\mu_{\hat{\boldsymbol{X}}}}^{*})\equiv\mathrm{sev}(\hat{\boldsymbol{X}}).
\end{equation}
This fact, together with optimality of $\hat{\boldsymbol{X}}_{\mu_{\hat{\boldsymbol{X}}}}^{*}$
for the Lagrangian relaxation (\ref{eq:Lagrangian}), implies
\begin{equation}
\mathrm{mse}(\hat{\boldsymbol{X}})+\mu_{\hat{\boldsymbol{X}}}\mathrm{sev}(\hat{\boldsymbol{X}})\ge\mathrm{mse}(\hat{\boldsymbol{X}}_{\mu_{\hat{\boldsymbol{X}}}}^{*})+\mu_{\hat{\boldsymbol{X}}}\mathrm{sev}(\hat{\boldsymbol{X}}_{\mu_{\hat{\boldsymbol{X}}}}^{*}),
\end{equation}
which further gives
\begin{equation}
\mathrm{mse}(\hat{\boldsymbol{X}})\ge\mathrm{mse}(\hat{\boldsymbol{X}}_{\mu_{\hat{\boldsymbol{X}}}}^{*}).
\end{equation}
Therefore, it is true that
\begin{align}
\mathrm{mse}(\hat{\boldsymbol{X}})\mathrm{sev}(\hat{\boldsymbol{X}}) & \ge\mathrm{mse}(\hat{\boldsymbol{X}}_{\mu_{\hat{\boldsymbol{X}}}}^{*})\mathrm{sev}(\hat{\boldsymbol{X}}_{\mu_{\hat{\boldsymbol{X}}}}^{*})\nonumber \\
 & \ge\inf_{\mu\in[0,\infty]}\mathrm{mse}(\hat{\boldsymbol{X}}_{\mu}^{*})\mathrm{sev}(\hat{\boldsymbol{X}}_{\mu}^{*}),
\end{align}
proving the claim  of the theorem in this case.

\textit{Case 2:} $\mathrm{sev}(\hat{\boldsymbol{X}})\equiv\mathrm{sev}(\hat{\boldsymbol{X}}_{\infty}^{*})$.
Because $\mathrm{sev}(\hat{\boldsymbol{X}})$ is convex
quadratic in $\hat{\boldsymbol{X}}$ and bounded below, it is fairly
easy to show that $\mathrm{sev}(\hat{\boldsymbol{X}}_{\infty}^{*})\equiv\inf_{\hat{\boldsymbol{X}}\in{\cal L}_{2|\mathscr{Y}}}\mathrm{sev}(\hat{\boldsymbol{X}})$.
Now, 
it either holds that
$
\mathrm{mse}(\hat{\boldsymbol{X}})\ge\mathrm{mse}(\hat{\boldsymbol{X}}_{\infty}^{*})\,(\equiv{\textstyle \lim_{\mu\uparrow\infty}}\mathrm{mse}(\hat{\boldsymbol{X}}_{\mu}^{*})),
$
giving
\begin{align}
\mathrm{mse}(\hat{\boldsymbol{X}})\mathrm{sev}(\hat{\boldsymbol{X}}) & \ge\mathrm{mse}(\hat{\boldsymbol{X}}_{\infty}^{*})\mathrm{sev}(\hat{\boldsymbol{X}}_{\infty}^{*}),\label{eq:C2_1}
\end{align}
or it must be true that
$
\mathrm{mse}(\hat{\boldsymbol{X}})<\mathrm{mse}(\hat{\boldsymbol{X}}_{\infty}^{*}).
$
In the latter case, the intermediate value property implies the existence
of a multiplier $\mu_{\hat{\boldsymbol{X}}}\in[0,\infty)$ such that
$\mathrm{mse}(\hat{\boldsymbol{X}})\equiv\mathrm{mse}(\hat{\boldsymbol{X}}_{\mu_{\hat{\boldsymbol{X}}}}^{*})$.
If $\mu_{\hat{\boldsymbol{X}}}>0$, optimality of $\hat{\boldsymbol{X}}_{\mu_{\hat{\boldsymbol{X}}}}^{*}$
for (\ref{eq:Lagrangian}) yields
\begin{equation}
\mathrm{mse}(\hat{\boldsymbol{X}})+\mu_{\hat{\boldsymbol{X}}}\mathrm{sev}(\hat{\boldsymbol{X}})\ge\mathrm{mse}(\hat{\boldsymbol{X}}_{\mu_{\hat{\boldsymbol{X}}}}^{*})+\mu_{\hat{\boldsymbol{X}}}\mathrm{sev}(\hat{\boldsymbol{X}}_{\mu_{\hat{\boldsymbol{X}}}}^{*}),
\end{equation}
or, equivalently,
$
\mathrm{sev}(\hat{\boldsymbol{X}})\ge\mathrm{sev}(\hat{\boldsymbol{X}}_{\mu_{\hat{\boldsymbol{X}}}}^{*})
$. Note that $\mathrm{sev}(\hat{\boldsymbol{X}})\equiv\mathrm{sev}(\hat{\boldsymbol{X}}_{\infty}^{*})$,
and so this actually implies that $\mathrm{sev}(\hat{\boldsymbol{X}}_{\infty}^{*})\equiv\mathrm{sev}(\hat{\boldsymbol{X}}_{\mu_{\hat{\boldsymbol{X}}}}^{*}).$
Regardless, we obtain
\begin{align}
\mathrm{mse}(\hat{\boldsymbol{X}})\mathrm{sev}(\hat{\boldsymbol{X}}) & \ge\mathrm{mse}(\hat{\boldsymbol{X}}_{\mu_{\hat{\boldsymbol{X}}}}^{*})\mathrm{sev}(\hat{\boldsymbol{X}}_{\mu_{\hat{\boldsymbol{X}}}}^{*}).\label{eq:C2_2}
\end{align}
If $\text{\ensuremath{\mu_{\hat{\boldsymbol{X}}}}\ensuremath{\ensuremath{\equiv}0}}$,
then $\hat{\boldsymbol{X}}\equiv\hat{\boldsymbol{X}}_{0}^{*}\equiv\mathbb{E}\{\boldsymbol{X}|\boldsymbol{Y}\}$
almost everywhere, which implies that
$
\mathrm{sev}(\hat{\boldsymbol{X}}_{\infty}^{*})\equiv\mathrm{sev}(\hat{\boldsymbol{X}})\equiv\mathrm{sev}(\hat{\boldsymbol{X}}_{0}),
$
and
\begin{equation}
\mathrm{mse}(\hat{\boldsymbol{X}})\mathrm{sev}(\hat{\boldsymbol{X}})\equiv\mathrm{mse}(\hat{\boldsymbol{X}}_{0}^{*})\mathrm{sev}(\hat{\boldsymbol{X}}_{0}^{*}).\label{eq:C2_3}
\end{equation}
From (\ref{eq:C2_1}), (\ref{eq:C2_2}) and (\ref{eq:C2_3}), we readily
see that
\begin{equation}
\mathrm{mse}(\hat{\boldsymbol{X}})\mathrm{sev}(\hat{\boldsymbol{X}})\ge\inf_{\mu\in[0,\infty]}\mathrm{mse}(\hat{\boldsymbol{X}}_{\mu}^{*})\mathrm{sev}(\hat{\boldsymbol{X}}_{\mu}^{*}),
\end{equation}
whenever $\hat{\boldsymbol{X}}$ is such that $\mathrm{sev}(\hat{\boldsymbol{X}})\equiv\mathrm{sev}(\hat{\boldsymbol{X}}_{\infty}^{*})$.

\textit{Case 3:} $\mathrm{sev}(\hat{\boldsymbol{X}})\notin[\mathrm{sev}(\hat{\boldsymbol{X}}_{\infty}^{*}),\mathrm{sev}(\hat{\boldsymbol{X}}_{0}^{*})]$.
Then we must necessarily have
\begin{equation}
\mathrm{sev}(\hat{\boldsymbol{X}})>\mathrm{sev}(\hat{\boldsymbol{X}}_{\mu}^{*}),\quad\forall\mu\in[0,\infty].
\end{equation}
In this case, either $\mathrm{mse}(\hat{\boldsymbol{X}})\equiv\mathrm{mse}(\hat{\boldsymbol{X}}_{\mu_{\hat{\boldsymbol{X}}}}^{*})$
for some $\mu_{\hat{\boldsymbol{X}}}\in[0,\infty]$, implying that
\begin{equation}
\mathrm{mse}(\hat{\boldsymbol{X}})\mathrm{sev}(\hat{\boldsymbol{X}})\ge\mathrm{mse}(\hat{\boldsymbol{X}}_{\mu_{\hat{\boldsymbol{X}}}}^{*})\mathrm{sev}(\hat{\boldsymbol{X}}_{\mu_{\hat{\boldsymbol{X}}}}^{*}),
\end{equation}
or $\mathrm{mse}(\hat{\boldsymbol{X}})>\mathrm{mse}(\hat{\boldsymbol{X}}_{\mu}^{*})$
for all $\mu\in[0,\infty]$, which gives
\begin{equation}
\mathrm{mse}(\hat{\boldsymbol{X}})\mathrm{sev}(\hat{\boldsymbol{X}})>\mathrm{mse}(\hat{\boldsymbol{X}}_{\mu}^{*})\mathrm{sev}(\hat{\boldsymbol{X}}_{\mu}^{*}),\,\,\forall\mu\in[0,\infty].
\end{equation}
Again, it follows that
\begin{equation}
\mathrm{mse}(\hat{\boldsymbol{X}})\mathrm{sev}(\hat{\boldsymbol{X}})\ge\inf_{\mu\in[0,\infty]}\mathrm{mse}(\hat{\boldsymbol{X}}_{\mu}^{*})\mathrm{sev}(\hat{\boldsymbol{X}}_{\mu}^{*}),
\end{equation}
and the proof is now complete.
\end{proof}
The practical aspects of Theorem \ref{UP} are summarized as follows: Provided an adequate threshold of $\mathrm{mse} (\mathrm{sev})$, the corresponding $\mathrm{sev} (\mathrm{mse})$ is always, at least, inversely proportional to that level.
Except for its analogy to classical uncertainty principles from physics and analysis (see Section \ref{intro}), our uncertainty relation resembles classical lower bounds in unconstrained/unbiased estimation, such as the Cram\`{e}r-Rao and Chapman-Robins bounds, in the sense that it provides a universal benchmark for any admissible estimator. Such an estimator might be chosen from the class of risk-aware estimators $\{\xemi\}_{\mu}$, or even from many other estimator classes (possibly more computationally friendly), such as linear estimators, deep neural networks, adaptive estimators, convex combinations of $\xeo$ and $\xeinf$, etc. However, under the setting of Theorem \ref{UP}, any estimator outside the family of risk-aware estimators $\{\xemi\}_{\mu}$ calls for Pareto improvement.
Further, estimators achieving the lower bound $\mathfrak{h}(\mathcal{P})$ must be equivalent to $\hat{\boldsymbol{X}}_{\mu^{\star}}^{*}$
(note that $\mu^{\star}$ might not be unique). 



\section{
Hedgeable Risk Margins
}\label{section4}
As expected, the Bayesian lower bound $\mathfrak{h}({\cal P})$ is achieved within the class of risk-aware estimators $\{\xemi\}_\mu$. In this section, we are interested in answering the following question:
Where are the estimators that achieve the lower bound with respect to the risk aversion parameter $\mu$ localized, and how does the width of such a localization area relate with the generative model  $\mathcal{P}_{(\boldsymbol{X},\boldsymbol{Y})}$? To this end, we now introduce a function that  measures the \textit{total projection} onto $\dx$ of the transformed risk-averse $\mu$-parameterized stochastic curve generated by $\xemi$ for all $\mu\ge0$, defined as
\begin{align}
\mathbb{C}(\boldsymbol{Y})&\triangleq{\int^{+\infty}_0 \Big\langle \cc^{\dagger}\frac{d\hat{\boldsymbol{X}}^{*}_{\tau}(\boldsymbol{Y})}{d\tau},\dx(\boldsymbol{Y})\Big\rangle~d\tau.\label{avproj}}
\end{align}
As we will see, $\mathbb{C}(\boldsymbol{Y})$ is actually nonnegative and expresses the \textit{margin to potentially counteract or hedge against risk}, as the latter is quantified by the $\mathrm{sev}$ functional; hereafter, we suggestively refer to  $\mathbb{C}(\boldsymbol{Y})$ as the \textit{hedgeable risk margin} associated with observation $\boldsymbol{Y}$.  By letting $\sminy$ be the smallest \textit{non-zero} eigenvalue of $\cc$, we have the following result. 
\begin{thm}[\textbf{Expected Hedgeable Risk Margin}]\label{risk-margin}
For fixed generative model $\mathcal{P}_{(\boldsymbol{X},\boldsymbol{Y})}$, $\mathbb{C}(\boldsymbol{Y})$ may be expressed as
\begin{equation}
    \mathbb{C}(\boldsymbol{Y})
    =
    \big\Vert\widehat{\boldsymbol{\Delta X}}(\boldsymbol{Y})\big\Vert_{\boldsymbol{\Sigma}_{\boldsymbol{X}|\boldsymbol{Y}}^{\dagger}}^{2} \ge 0,
\end{equation}
and its expected value satisfies the standard bounds
\begin{equation}
{
\mathbb{E}
\big\{
\mathcal{E}^{2}_L( \boldsymbol{Y})
\big\}}\le{\mathbb{E}\{\mathbb{C}(\boldsymbol{Y})\}}\leq{
\mathbb{E}
\big\{
\mathcal{E}^{2}_U( \boldsymbol{Y})
\big\}},\label{eq28}
\end{equation}
where 
\begin{equation}
\mathcal{E}_{U(L)}( \boldsymbol{Y})\triangleq
\begin{cases}
\dfrac{\big\|\dx(\boldsymbol{Y})\big\|_{2}}{\sqrt{\sigma_{\min(\max)}{(\boldsymbol{Y})}}}, & \text{if }\sminy>0\\
0, & \text{if not}
\end{cases}.
\end{equation}
\end{thm}
\begin{proof}[Proof of Theorem 2]
In case $\sminy=0$ (which happens if and only if $\smaxy=0$), the situation is trivial and the result holds. Therefore, in what follows we may assume that $\sminy>0$. In that case, to obtain an expression for \eqref{avproj} we have to differentiate \eqref{eq:RA_MMSE} with respect to $\mu$ which for brevity may be  written as
$
\xemi=\boldsymbol{\zeta}(\mu)(\xeo+\mu \boldsymbol{R}),\label{eq39}
$
where 
$
\boldsymbol{\zeta}(\mu)\triangleq(\boldsymbol{I}+2\mu\cc)^{-1}~.
$
By fixing an observation $\boldsymbol{Y}$,
we obtain the linear $\mu$-varying system 
\begin{equation}\label{eq50}
\begin{aligned}
\frac{d\xemi}{d\mu}=-2\boldsymbol{\zeta}(\mu)\cc\xemi+\boldsymbol{\zeta}(\mu)\boldsymbol{R}.
\end{aligned}
\end{equation}
Then, given that the commutator
$
\big[ \boldsymbol{\zeta}(\mu), \cc\big]=\boldsymbol{0}
$,
\eqref{eq50} can be written as
\begin{equation}
\begin{aligned}
\frac{d\xemi}{d\mu}=\boldsymbol{\zeta}(\mu)^2(\boldsymbol{R}-2\cc\xeo).\label{eq34}
\end{aligned}
\end{equation}
Now, from \eqref{eq6} we have
\begin{equation}
    \xeinf=\boldsymbol{U}\boldsymbol{K}\boldsymbol{U}^{\top}\xeo+\frac{1}{2}\boldsymbol{U}\boldsymbol{D}_{\cc}^{\dagger}\boldsymbol{U}^{\top}\boldsymbol{R},
\end{equation}
where  
\begin{equation}
\boldsymbol{D}_{\cc}^{\dagger}=
\mathrm{diag}\big(\big\{(\siy)^{-1}\big\}_{i\in{\mathbb{N}^+_r}},\boldsymbol{0}\big) 
\end{equation}
and  
\begin{equation}
\boldsymbol{K}=\mathrm{diag}\big(\big\{0\}_{i\in{\mathbb{N}^+_r}},\boldsymbol{1}\big).
\end{equation}
Thus, 
\begin{align}
\boldsymbol{U}^{\top}\boldsymbol{R}
&=\boldsymbol{D}_{\cc}\big( 2\boldsymbol{U}^{\top}\xeinf-2\boldsymbol{K}\boldsymbol{U}^{\top}\xeo\big)\nonumber\\[5pt]
&=2\boldsymbol{D}_{\cc}\boldsymbol{U}^{\top} \xeinf.\label{fak}
\end{align}
From \eqref{eq34}, and \eqref{fak},  we have:
\begin{align}
\frac{d\xemi(\boldsymbol{Y})}{d\mu}
&=2\boldsymbol{U}\boldsymbol{\Lambda}(\boldsymbol{Y})^2\boldsymbol{D}_{\cc}\boldsymbol{U}^{\top}\dx,\label{eq37}
\vspace{-0.3cm}
\end{align}
where  
\vspace{-0.1cm}
\begin{align}
 \boldsymbol{\Lambda}(\boldsymbol{Y})^2=
\mathrm{diag}\big( \big\{{(1+2\mu\siy)^{-2}}\big\}_{i\in{\mathbb{N}^+_r}},\boldsymbol{1}\big).
\end{align}
Therefore, the integrand reads:
\begin{align}
&\hspace{-0.25cm}\Big\langle \cc^{\dagger}\frac{d\hat{\boldsymbol{X}}^{*}_{\mu}(\boldsymbol{Y})}{d\mu},\dx\Big\rangle \nonumber\\[5pt]
&=2\dx^{\top}\boldsymbol{U}\boldsymbol{\Lambda}(\boldsymbol{Y})^2\boldsymbol{(D^{\dagger}_{\cc})}\boldsymbol{D}_{\cc}\boldsymbol{U}^{\top}\dx, \label{eq40}
\end{align}
from which it follows that
\begin{align}
\mathbb{C}(\boldsymbol{Y})=[\boldsymbol{U}^{\top}\dx]^{\top}\boldsymbol{D^{\dagger}_{\cc}}\boldsymbol{U}^{\top}\dx.
\end{align}
Thus, provided the assumptions from \cite{kalogerias2020better} we obtain
\begin{align}
\mathbb{E}\bigg\{\frac{\big\|\dx\|^2_2}{{\smaxy}}\bigg\} \leq\mathbb{E}\{\mathbb{C}(\boldsymbol{Y})\}\leq \mathbb{E}\bigg\{\frac{\big\|\dx\|^2_2}{{\sminy}}\bigg\},\label{eq43}
\end{align}
and we are done.
\end{proof}
At this point it is worth attributing geometric meaning in the above result; by integrating \eqref{eq37} in $(0,\mu)$ we obtain:
\begin{align}
\xemi(\boldsymbol{Y})=\xeo(\boldsymbol{Y})+\boldsymbol{U}\boldsymbol{G}(\mu) \boldsymbol{U}^{\top}\dx, \label{eq41}
\vspace{-1.5cm}
\end{align}
where 
\begin{align}
 \boldsymbol{G}(\mu)=   \mathrm{diag}\big( \big\{{2\mu\siy(1+2\mu\siy)^{-1}}\big\}_{i\in{\mathbb{N}^+_r}},\boldsymbol{0}\big).
\end{align}
We observe that the risk-aware estimator shifts the conditional mean estimator by the transformed difference $\dx$.
Thus, motivated by the one dimensional case we may interpret $\dx$ as the direction of asymmetry of the posterior (for the given observation), and note the following: referring to \eqref{eq40}, $[\boldsymbol{U}^{\top}\dx]_{i}$ being large enough for most of ${i\in\mathbb{N}_n^{+}}$ implies that large estimation errors incurred by the conditional mean estimator are mostly due to the built-in riskiness of the posterior. In this case, the projection from \eqref{eq40} decreases with $\mu$ over a long width before fading-out. 
\par
To put it differently, large projections indicate enough margin with respect to $\mu$ to potentially hedge against risk, justifying the meaning ascribed in $\mathbb{C}(\boldsymbol{Y})$.
Inequality \eqref{eq43} implies that, on average,
the information regarding the \textit{active} risk-aware estimates -and subsequently those that achieve the lower bound- is completely embodied to the limit points of the curve. Thus, recalling \eqref{SkS} and \eqref{sufcond}, we expect that "near" a skew symmetric generative model those risk-averse estimates which actively account for risk will be limited. Further, highly skewed models compress the active risk-aversion range. To see that consider $||\dx||^2_2\neq0$ and then take
\begin{align}
\hspace{-0.25cm}\Big\langle \cc^{\dagger}\dfrac{d\hat{\boldsymbol{X}}^{*}_{\mu}(\boldsymbol{Y})}{d\mu},\dfrac{\dx}{||\dx||_{2}}\Big\rangle&=2\dx^{\top}\boldsymbol{U}\boldsymbol{\Lambda}(\boldsymbol{Y})^2\boldsymbol{(D^{\dagger}_{\cc})}\boldsymbol{D}_{\cc}\boldsymbol{U}^{\top}\dfrac{\dx}{||\dx||_{2}}\nonumber\\[5pt]
&=\dfrac{2}{||\dx||_{2}}\sum^{r}_{i=1}\dfrac{1}{(1+2\mu\sigma_{i})^{2}}[\boldsymbol{U}^{\top}\dx]_{i}^{2}\nonumber\\[5pt]
&<\dfrac{1}{\mu\rho_{\min}||\dx||_{2}}\sum^{r}_{i=1}[\boldsymbol{U}^{\top}\dx]_{i}^{2}.
\end{align}
Thus, by choosing $\epsilon>0$ we may write
\begin{align}
    \mu^{*}<\dfrac{\sqrt{\mathbb{E}\{\sum^{r}_{i=1}[\boldsymbol{U}^{\top}\dx]_{i}^{2} \}}}{\epsilon\rho_{\min}\sqrt{\mathbb{E}\{||\dx||^2_2\}}}+\mathcal{O}\Bigg(\frac{\mathrm{Var}({||\dx||^2_2})}{\epsilon({\mathbb{E}\{||\dx||^2_2\}})^{\frac{3}{2}}}\Bigg),~\text{with}~ \frac{||\dx||^2_2-\mathbb{E}\{||\dx||^2_2\}}{\mathbb{E}\{||\dx||^2_2\}}\rightarrow0,
\end{align}
where $\rho_{\mathrm{min}}$ is such that $\textup{ess}\hspace{1bp}\textup{inf}\,\sminy\geq\rho_{\min}$. 
\section{Lower Bound Characterization }\label{section5}
As we saw earlier
, provided a generative model, there exist risk-aware estimators that result in both good (even optimal) performance on average, and an adequate level of robustness; a standard example is the efficient frontier family $\{\hat{\boldsymbol{X}}_{\mu}^{*}\}_\mu$, and in particular for $\mu\equiv\mu^*$ (see Theorem \ref{UP}). But still, how far can the trade-off incurred by any member of the efficient frontier class $\{\hat{\boldsymbol{X}}_{\mu}^{*}\}_\mu$ be from achieving the ultimate lower bound $\mathrm{mse}(\hat{\boldsymbol{X}}_{0}^{*})\mathrm{sev}(\hat{\boldsymbol{X}}_{\infty}^{*})$, and how is this distance related to the assumed generative model? We answer these questions by showing that the difference between the parameterization $\mathrm{mse}(\xemi)\mathrm{sev}(\xemi)$ and $\mathrm{mse}(\hat{\boldsymbol{X}}_{0}^{*})\mathrm{sev}(\hat{\boldsymbol{X}}_{\infty}^{*})$ is bounded from above \textit{and} below by functions of another positive, risk margin-related, model-dependent functional.

\begin{thm}[\textbf{Uncertainty Bound Characterization}]\label{Bound}
Suppose that there exists $\rho_{\max}\ge0$, such that $\textup{ess}\hspace{1bp}\textup{sup}\,\smaxy\leq\rho_{\max}$. Then, the products $\mathrm{mse}(\xemi)\mathrm{sev}(\xemi)$, $\mu\in[0,\infty]$
and 
$\mathrm{mse}(\xeo)\mathrm{sev}(\xeinf)$ satisfy the uniform upper bound 
\begin{equation}
    \mathrm{mse}(\xemi)\mathrm{sev}(\xemi)\hspace{-1bp}-\hspace{-1bp}\mathrm{mse}(\xeo)\mathrm{sev}(\xeinf)
    \hspace{-1bp}\leq\hspace{-1bp}
    \mathbb{U}(\mathcal{P}),
\end{equation}
where 
\begin{align}
\mathbb{U}(\mathcal{P})&=
\big((\rho_{\max})^{2}\mathrm{mse}(\xeo) + \rho_{\max}\mathrm{sev}(\xeinf)\big)d(\mathcal{P})^2
+(\rho_{\max})^{3}d(\mathcal{P})^4,
\end{align}
and $d({\cal P}) \triangleq2\sqrt{\mathbb{E}\{\mathbb{C}(\boldsymbol{Y})\}}$.
If, further, there exists $\rho_{\min}>0$, such that $\textup{ess}\hspace{1bp}\textup{inf}\,\sminy\geq\rho_{\min}$, then  the same products 
satisfy the lower bound
\begin{equation}
    \mathbb{L}(\mathcal{P},\mu)\hspace{-1bp}\le\hspace{-1bp}
    \mathrm{mse}(\xemi)\mathrm{sev}(\xemi)-\mathrm{mse}(\xeo)\mathrm{sev}(\xeinf)
    \hspace{-1bp}
\end{equation}
where
\begin{align}
\mathbb{L}(\mathcal{P},\mu)&=
\big(\alpha(\mu)\mathrm{mse}(\xeo) + \rho_{\min}\mu^2\alpha(\mu)\mathrm{sev}(\xeinf)\big)d(\mathcal{P})^2
+(\rho_{\min})\mu^2\alpha(\mu)^2d(\mathcal{P})^4,
\end{align}
and $ \mathlarger{\alpha(\mu)=(1/4){{\rho^2_{\min}}{(1+2\mu\rho_{\max})^{-2}}}}$.
\end{thm}
\begin{proof}[Proof of Theorem 3]
To begin with, under the setting of the theorem, let us integrate \eqref{eq37} in $(\mu,\mu')$, obtaining
\begin{equation}\label{eq46}
\begin{aligned}
\xemi-\xemit=(\mu-\mu')\boldsymbol{U}\boldsymbol{H}(\mu,\mu')\boldsymbol{U}^{\top}\dx,
\end{aligned}
\end{equation}
where 
 \begin{align}
 &\hspace{-4bp}\boldsymbol{H}(\mu,\mu')  
=\mathrm{diag}\Bigg( \bigg\{\frac{2\siy}{\big(1+2\mu\siy\big)\big(1+2\mu'\siy\big)}\bigg\}_{i\in{\mathbb{N}^+_r}},\boldsymbol{0}\Bigg). \nonumber
\end{align}
Subsequently, consider the difference
\begin{align}
&\hspace{-0.1cm}|\mathrm{mse}(\xemi)-\mathrm{mse}(\xemit)|\nonumber
=\big|\mathbb{E}\big\{(\xemi-\xeo+\xemit-\xeo)^{\top}(\xemi-\xemit)\big\}\big|.
\end{align}
After substituting $\mu'=0$ and subsequently applying \eqref{eq46} and \textit{Lemma 2}, we get 
\begin{align}
\Lambda_{\mathrm{mse}}(\mu)
&\triangleq\mathrm{mse}(\hat{\boldsymbol{X}}_{\mu}^{*})-\mathrm{mse}(\hat{\boldsymbol{X}}_{0})\nonumber\\
&=\mathbb{E}\big\{(\hat{\boldsymbol{X}}_{\mu}^{*}-\xeo)^{\top}(\hat{\boldsymbol{X}}_{\mu}^{*}-\xeo)\big\}.
\end{align}
Additionally, recalling the QCQP reformulation of the $\mathrm{sev}$-constrained MMSE estimation problem in \cite{kalogerias2020better}, we may write
\begin{align}
&\hspace{-0cm}|\mathrm{sev}(\hat{\boldsymbol{X}}_\mu)-\mathrm{sev}(\hat{\boldsymbol{X}}_{\mu'})| \label{ok} \\
&=\big|\mathbb{E}\{ (\xemi-\xeinf+\xemit-\xeinf)^{\top}\cc(\xemi-\xemit)\}\big|. \nonumber
\end{align}
Thus, by substituting $\mu'=+\infty$, \eqref{ok} yields 
\begin{align}
\Lambda_{\mathrm{sev}}(\mu)
&\triangleq\mathrm{sev}(\hat{\boldsymbol{X}}_{\mu}^{*})-\mathrm{sev}(\hat{\boldsymbol{X}}^{*}_{\infty})\nonumber\\
&=\mathbb{E}\big\{ (\hat{\boldsymbol{X}}_{\mu}^{*}-\xeinf)^{\top}\cc(\hat{\boldsymbol{X}}_{\mu}^{*}-\xeinf)\big\}.
\end{align}
From \eqref{eq46}, Lemma \ref{lemma2} and Theorem \ref{risk-margin}, it is easy to show that
\begin{equation}
 \Lambda_{\mathrm{mse}}(\mu)\leq\frac{\rho_{\max}}{4}d(\mathcal{P})^{2} \label{eq60}
 \,\,\text{and}\,\,
 \Lambda_{\mathrm{sev}}(\mu)\leq\frac{\rho_{\max}^2}{4}d(\mathcal{P})^{2},
\end{equation}
Therefore, from \eqref{eq60}, we may write
\begin{align}
\mathrm{mse}(\hat{\boldsymbol{X}}_{\mu}^{*})\mathrm{sev}(\hat{\boldsymbol{X}}^{*}_{\mu})- \mathrm{mse}(\xeo)\mathrm{sev}(\xeinf)
&=\Lambda_{\mathrm{sev}}(\mu)\mathrm{mse}(\xeo)\nonumber
+\Lambda_{\mathrm{mse}}(\mu)\mathrm{sev}(\xeinf)
+\Lambda_{\mathrm{mse}}(\mu)\Lambda_{\mathrm{sev}}(\mu)
\nonumber\\
&\leq\mathbb{U}(\mathcal{P}),\quad \forall \mu \in[0,+\infty].\label{eq54}
\end{align}
Lastly, when a $\rho_{\min}$ exists, again from \eqref{eq46} and Theorem \ref{risk-margin} we may also fairly easily find that
\begin{align}
 \Lambda_{\mathrm{mse}}(\mu)
 &\geq\rho_{\min}\mu^2\alpha(\mu)d(\mathcal{P})^{2}\label{eq61}
 \quad\text{and}
\\
\Lambda_{\mathrm{sev}}(\mu)
&\geq \alpha(\mu)d(\mathcal{P})^{2},
\end{align}
and thus in a similar manner obtain the lower bound $\mathbb{L}(\mathcal{P},\mu)$.
Enough said. \end{proof}
Theorem \ref{Bound} implies that for sufficiently small $\varepsilon>0$ for which $d(\mathcal{P})<\varepsilon$,
it is true that, \textit{uniformly} over $\mu \in[0,+\infty]$, 
\begin{align}
 \mathrm{mse}(\hat{\boldsymbol{X}}_{\mu}^{*})\mathrm{sev}(\hat{\boldsymbol{X}}^{*}_{\mu})
\simeq \mathfrak{h}(\mathcal{P})\simeq \mathrm{mse}(\xeo)\mathrm{sev}(\xeinf)\label{eq55}.
\end{align}
In other words, when $d(\mathcal{P})$ is very small, we can select the risk aversion parameter $\mu$ almost freely and still achieve simultaneously both a good average performance and an adequate level of robustness; this is of course a feature of (near-)skew-symmetric  models.
%
%
On the contrary, highly skewed models displace the optimal trade-off $\mathfrak{h}(\mathcal{P})$ away from the ultimate lower bound, thus rendering the exchangeability between $\mathrm{mse}$ and $\mathrm{sev}$  highly nontrivial. Given fixed values of $\rho_{\min}$ and $\rho_{\max}$, Theorem \ref{Bound} also implies that the optimal trade-off $\mathfrak{h}(\mathcal{P})$ is fully characterized by three numbers:  $d(\mathcal{P})$, the minimum $\mathrm{mse}$ and the minimum $\mathrm{sev}$. 
Next, we show that $d(\mathcal{P})$ admits simultaneously well-defined and intuitive topological \textit{and} statistical interpretations, within a rigorous framework.


\section{Risk Margins as Complete Metrics and Measures of Skewness in High Dimensions}\label{section6}

In what follows, denote the product of state and observable spaces
as $S\triangleq\mathbb{R}^{n}\times\mathbb{R}^{m}$, and let $\mathsf{P}(S)$
be the set of all Borel probability measures on $S$. Also recall
the risk margin-related functional $d:\mathsf{P}(S)\rightarrow\mathbb{R}_{+}$
defined in Theorem \ref{Bound} as
\begin{align}
d({\cal P}) & \equiv2\sqrt{\mathbb{E}_{{\cal P}_{\boldsymbol{Y}}}\{\mathbb{C}({\cal P}_{\boldsymbol{X}|\boldsymbol{Y}})\}}\nonumber \\
 & =2\sqrt{\mathbb{E}_{{\cal P}_{\boldsymbol{Y}}}\Big\{\big\Vert\widehat{\boldsymbol{\Delta X}}({\cal P}_{\boldsymbol{X}|\boldsymbol{Y}})\big\Vert_{\boldsymbol{\Sigma}_{\boldsymbol{X}|\boldsymbol{Y}}^{\dagger}}^{2}\Big\}},
\end{align}
where we now explicitly highlight the dependence on the generative
model ${\cal P}\equiv{\cal P}_{(\boldsymbol{X},\boldsymbol{Y})}$.
Then, we consider the space 
\begin{equation}
\mathsf{P}_{\mathbb{S}}(S)\triangleq\{{\cal P}\in\mathsf{P}(S)|d({\cal P})<\infty\},
\end{equation}
as well as the feasibility set
\begin{equation}
{\cal F}\triangleq\{\alpha\ge0|d({\cal P})=\alpha,\text{for some }{\cal P}\in\mathsf{P}_{\mathbb{S}}(S)\}\subseteq\mathbb{R}_{+}.
\end{equation}

Our discussion will concentrate on endowing $\mathsf{P}_{\mathbb{S}}(S)$
with a topological structure based on appropriate handling of the
functional $d$, resulting among other things in a meaningful and intuitive
topological interpretation for the latter.

Indeed, for every number $\alpha\in{\cal F}$, take an arbitrary element
${\cal P}_{\alpha}\in\mathsf{P}_{\mathbb{S}}(S)$ such that $d({\cal P}_{\alpha})=\alpha$.
Then, we may construct a measure-valued multifunction ${\cal C}:{\cal F}\rightrightarrows\mathsf{P}_{\mathbb{S}}(S)$
as
\begin{equation}
{\cal C}(\alpha)\triangleq\{{\cal P}\in\mathsf{P}(S)|{\cal P}\sim{\cal P}_{\alpha}\},
\end{equation}
which in turn defines an equivalence class of ${\cal P}_{\alpha}$
for each $\alpha\in{\cal F}$, as well as consider \textit{any} selection
of the multifunction ${\cal C}$, say $C:{\cal F}\rightarrow\mathsf{P}_{\mathbb{S}}(S)$,
i.e., a measure-valued function such that $C(\cdot)\in{\cal C}(\cdot)$
on ${\cal F}$. Next, we define a set of \textit{equivalence class
representatives} as
\begin{equation}
{\cal R}=\mathrm{range}(C).
\end{equation}
Based on our construction, one may always choose $C(\cdot)=\text{\ensuremath{{\cal P}_{(\cdot)}}}$
on ${\cal F}$, in which case ${\cal R}=\{{\cal P}_{\alpha}\}_{\alpha\in{\cal F}}$.
There is a bijective mapping between ${\cal R}$ and the collection
of equivalence classes $\{{\cal C}(\alpha)\}_{\alpha\in{\cal F}}$.
Therefore, we may define the \textit{canonical projection map}
\begin{align}
\Pi({\cal P}) & =\underset{\tilde{{\cal P}}\in{\cal R}}{\arg\min}\,|d({\cal P})-d(\tilde{{\cal P}})|\in{\cal R},
\end{align}
which maps every Borel measure ${\cal P}$ in $\mathsf{P}_{\mathbb{S}}(S)$
to its representative $\Pi({\cal P})$ in ${\cal R}$ and equivalently,
to its corresponding equivalence class. In other words, the canonical
map $\Pi$ \textit{separates} or \textit{partitions} $\mathsf{P}_{\mathbb{S}}(S)$
on the basis of the values of the risk margin statistic $d({\cal P})$, for each
${\cal P}\in\mathsf{P}_{\mathbb{S}}(S)$.

Let us now define another related functional $d_{\mathbb{S}}:\mathsf{P}_{\mathbb{S}}(S)\times\mathsf{P}_{\mathbb{S}}(S)\rightarrow\mathbb{R}_{+}$
as
\begin{align}
\hspace{-6bp}d_{\mathbb{S}}({\cal P},{\cal P}') &	\triangleq\sqrt{\big|\mathbb{E}_{{\cal P}_{\boldsymbol{Y}}}\{\mathbb{C}({\cal P}_{\boldsymbol{X}|\boldsymbol{Y}})\}-\mathbb{E}_{{\cal P}'_{\boldsymbol{Y}}}\{\mathbb{C}({\cal P}'_{\boldsymbol{X}|\boldsymbol{Y}})\}\big|}\nonumber
	\\&=\sqrt{\big|(d({\cal P}))^{2}-(d({\cal P}'))^{2}\big|}.
\end{align}
It is easy to see that the pair $({\cal R},d_{\mathbb{S}})$ is a metric
space. In fact, it is immediate that, for every ${\cal P}\in{\cal R}$
and ${\cal P}'\in{\cal R}$,
\begin{equation}
d_{\mathbb{S}}({\cal P},{\cal P}')=0\iff d({\cal P})=d({\cal P}')\iff{\cal P}\equiv{\cal P}',
\end{equation}
$d_{\mathbb{S}}({\cal P},{\cal P}')=d_{\mathbb{S}}({\cal P}',{\cal P})$,
and for another ${\cal P}''\in{\cal R}$
\begin{align}
d_{\mathbb{S}}({\cal P},{\cal P}'')  & =\sqrt{\big|(d({\cal P}))^2-(d({\cal P}'))^2+(d({\cal P}'))^2-(d({\cal P}''))^2\big|}\nonumber \\
 & \le\sqrt{\big|(d({\cal P}))^2-(d({\cal P}'))^2\big|}+\sqrt{\big|(d({\cal P}'))^2-(d({\cal P}''))^2}\big|\nonumber \\
 & =d_{\mathbb{S}}({\cal P},{\cal P}')+d_{\mathbb{S}}({\cal P}',{\cal P}'').
\end{align}
Since $({\cal R},d_{\mathbb{S}})$ is indeed a metric space,
$d_{\mathbb{S}}$
induces a topology on the representative set ${\cal R}$, which we
suggestively call the \textit{(hidden) skewed topology} on ${\cal R}$.
Similarly, $(\mathsf{P}_{\mathbb{S}}(S),d_{\mathbb{S}})$ is a pseudometric
space. 
In this case, we say that $d_{\mathbb{S}}$ induces the \textit{(hidden)
skewed pseudometric topology }on $\mathsf{P}_{\mathbb{S}}(S)$. In fact, we may prove more.
\begin{thm}
The metric space $({\cal R},d_{\mathbb{S}})$ is Polish, and the pseudometric space $(\mathsf{P}_{\mathbb{S}}(S),d_{\mathbb{S}})$
is pseudoPolish.
\end{thm}
Before stating the proof of the theorem, we provide the following Lemma:
\begin{lem}[] The functional $d$ is surjective, i.e., $\mathcal{F}=\mathbb{R}_{+}$.
\end{lem}

\begin{proof}[Proof of Lemma 3]
Consider the simple case of a totally hidden state vector $\boldsymbol{X}$ and suppose that 
\begin{equation}
\mathcal{P}_{(\boldsymbol{X},\boldsymbol{Y})}=\mathcal{P}_{(\boldsymbol{X})}=\mathlarger{\Pi}_{i=1}^{n}\mathcal{P}_{({X_i})}~,
\end{equation}
where $\mathcal{P}_{(X_i)} \sim \mathrm{gamma}(\kappa_i,\theta_i)$. Then 
\begin{align}
d(\mathcal{P_{\boldsymbol{X}}})
&=\Vert\dx(\mathcal{P}_{\boldsymbol{X}})\Vert_{\boldsymbol{\Sigma}^{\dagger}_{\boldsymbol{X}}}\nonumber\\[5pt]
&=2\left\Vert\left[\begin{array}{c}\left[
\mathbb{E}\{X_i \} -\dfrac{\mathbb{E}\{X^{3}_i\}-\mathbb{E}\{X^{2}_i\}\mathbb{E}\{X_i\} \}}{2\sigma^{2}_i}\right]^{1}_{r}\\
\boldsymbol{0}_{n-r}
\end{array}\right] \right\Vert_{\boldsymbol{\Sigma}^{\dagger}_{\boldsymbol{X}}}.\label{katara}
\end{align}
Since the skewness of $\mathcal{P}_{(X_i)}$ is given by $\frac{2}{\sqrt{\kappa_i}}$, \eqref{katara} reads
\begin{equation}
d(\mathcal{P_{\boldsymbol{X}}})=4\sqrt{\frac{1}{\kappa_{i}\sigma^2_{i}}}~.
\end{equation}
As a result, for any $\alpha\geq0$, we can always find $\kappa_{i}>0$ and $\sigma_{i}=\kappa_i\theta^2_i>0$ such that 
\begin{equation}
d(\mathcal{P_{\boldsymbol{X}}})=4\sqrt{\frac{1}{\kappa^3_{i}\theta^4_{i}}}=\alpha.
\end{equation}
In other words, for any $\alpha\geq0$,  $d(\mathcal{P})=\alpha$ has always a solution within $\mathsf{P}_{\mathbb{S}}(S)$ which concludes the proof.
\end{proof}
\begin{proof}[Proof of Theorem 4]
In order to prove separability for $(\mathcal{R},d_{\mathbb{S}})$, take any $\mathcal{P}_{0} \in \mathcal{R}$ and its corresponding image through $d$, $d(\mathcal{P}_0)=\alpha_0$. Since $\mathcal{F}$ is separable, there exist rationals $\{\alpha_{n}\}^{\infty}_{n=1}$ s.t. $\alpha_n\rightarrow\alpha_{0}$. Further,  due to \textit{Lemma 3}, there exists $\{\mathcal{P}_{n}\}^{\infty}_{n=1} \in \mathcal{R}$ with $d(\mathcal{P}_{n})=\alpha_n$. Thus,
\begin{align}
\alpha_{n}\rightarrow \alpha_{0}   &\Rightarrow d(\mathcal{P}_{n})\rightarrow d(\mathcal{P}_0) \nonumber\\[5pt]
&\Rightarrow d(\mathcal{P}_{n})^{2}\rightarrow d(\mathcal{P}_0)^{2}\nonumber\\[5pt]
& \Rightarrow \sqrt{|d(\mathcal{P}_{n})^{2}-d(\mathcal{P}_0)^{2} |} \rightarrow 0 \nonumber\\[5pt]
&\Rightarrow d_{\mathbb{S}}(\mathcal{P}_n,\mathcal{P}_0)\rightarrow 0 \label{tsovolas}
\end{align}
To show that $(\mathcal{R},d_{\mathbb{S}})$ is complete, take $\{\mathcal{P}_{n}\}^{\infty}_{n=1} \in \mathcal{R}$ to be Cauchy. Then, since 
$
(\forall \varepsilon>0) (\exists N=N(\varepsilon)>0)~$ s.t. 
\begin{align}
  n,m>N(\varepsilon)&\Rightarrow d_{\mathbb{S}}(\mathcal{P},\mathcal{P}')<\varepsilon \nonumber\\[5pt]
  &\Rightarrow \sqrt{|d(\mathcal{P}_{n})^{2}-d(\mathcal{P}_0)^{2} |}<\varepsilon \nonumber\\[5pt]
  &\Rightarrow \sqrt{|{\alpha}_{n}^{2}-\alpha_m^{2} |}<\varepsilon~,
\end{align}
 $\{\alpha_{n}\}^{\infty}_{n=1} \in \mathcal{F}$ is Cauchy and hence it converges to some $\alpha_0 \in \mathcal{F}$. Since $d$ is onto, the same process followed in \eqref{tsovolas} yields $\mathcal{P}_{n} \rightarrow \mathcal{P}_0$.
For the metric space $(\mathsf{P}(S),d_{\mathbb{S}})$, separability follows by fixing a collection of representatives $\mathcal{R}$ and subsequently choosing $\mathcal{P}_0\in \mathcal{R}$. Then, a countable basis for $(\mathcal{R},d_{\mathbb{S}})$ also separates $(\mathsf{P}(S),d_{\mathbb{S}})$ since for any $\mathcal{P}^{*}\sim\mathcal{P}_0$ we have 
 \begin{align}
d_{\mathbb{S}}(\mathcal{P}_n,\mathcal{P}^{*}) &\leq d_{\mathbb{S}}(\mathcal{P}_n,\mathcal{P}_{0})+d_{\mathbb{S}}(\mathcal{P}_0,\mathcal{P}^{*})\nonumber\\[5pt]
&\leq d_{\mathbb{S}}(\mathcal{P}_n,\mathcal{P}_{0})+0 \label{diki}
\end{align}
which implies that $\mathcal{P}_{n} \rightarrow \mathcal{P}^{*}$. Lastly, given that $\mathcal{P}_0$ is arbitrary, \eqref{diki} shows completeness for $(\mathsf{P}(S),d_{\mathbb{S}})$ since convergence of a Cauchy sequence in $(\mathcal{R},d_{\mathbb{S}})$ extends to convergence to an equivalence class in $(\mathsf{P}(S),d_{\mathbb{S}})$.\end{proof}

Therefore, there is a standard topological structure induced by $d_{\mathbb{S}}$
(and thus by $d$) on $\mathsf{P}_{\mathbb{S}}(S)$ with a complete
description and favorable properties in terms of separation, closeness
and limit point behavior. 

Under these structural considerations, it is then immediate to observe
that for any given Borel measure ${\cal P}\in\mathsf{P}_{\mathbb{S}}(S)$
we have that $d({\cal P})=|(d(\Pi({\cal P})))^2-0|^{1/2}=d_{\mathbb{S}}(\Pi({\cal P}),{\cal P}_{0})$,
and therefore we may \textit{interpret $d({\cal P})$ as the distance
of ${\cal P}$ relative to all equivalent to each other skew-symmetric
Borel measures on $S$} (i.e., with $\Vert\widehat{\boldsymbol{\Delta X}}\Vert_{2}=0$
almost everywhere), which are precisely the measures for which risk-neutral
and risk-aware estimators, $\hat{\boldsymbol{X}}_{0}$ and $\hat{\boldsymbol{X}}_{\infty}$
respectively, coincide and thus the corresponding $\textrm{mse}$
and $\textrm{sev}$ are simultaneously minimal. 
This fact is significant, not only because it provides a clear topological meaning for the (expected) risk margin analyzed earlier in Section \ref{section4} (Theorem \ref{risk-margin}), but also because $d(\mathcal{P})$ 
induces a similar interpretation to the optimal trade-off $\mathfrak{h}(\mathcal{P})$ via Theorem \ref{Bound}, and consequently completely characterizes the general $\mathrm{mse}$/$\mathrm{sev}$ trade-off of the uncertainty principle of Theorem \ref{UP}.

Simultaneously, both functionals $d$ and $d_\mathbb{S}$ admit a convenient and intuitive \textit{statistical interpretation}, as well. To see this, let us consider the simplest case of a totally hidden, real-valued state variable, say $X$. We have $\mathcal{P}\equiv{\mathcal{P}_{X|\mathbf{0}}}\equiv{\mathcal{P}_X}$, where $\mathbf{0}$ denotes a fictitious trivial observation. Then, denoting the mean and variance of $X$ as $\mu$ and $\sigma^2$, respectively, $d(\mathcal{P})$ may be expressed as
\begin{align}
d({\cal P}) & =2\dfrac{1}{\sigma}\big|\widehat{\boldsymbol{\Delta}X}({\cal P}_{X})\big|\nonumber \\
 & =2\dfrac{1}{\sigma}\bigg|\mathbb{E}\{X\}-\dfrac{\mathbb{E}\{X^{3}\}-\mathbb{E}\{X^{2}\}\mathbb{E}\{X\}}{2\sigma^{2}}\bigg|\nonumber \\
 & =2\dfrac{1}{\sigma}\bigg|\dfrac{2\sigma^{2}\mu-\mathbb{E}\{X^{3}\}+\mathbb{E}\{X^{2}\}\mu}{2\sigma^{2}}\bigg|\nonumber \\
 & =\bigg|\dfrac{2\sigma^{2}\mu-\mathbb{E}\{X^{3}\}+(\mu^{2}+\sigma^{2})\mu}{\sigma^{3}}\bigg|\nonumber \\
 & =\bigg|\dfrac{-\mathbb{E}\{X^{3}\}+\mu^{3}+3\sigma^{2}\mu}{\sigma^{3}}\bigg|,
\end{align}
or, equivalently,
\begin{equation}
    d({\cal P})
    =\bigg|\mathbb{E}\bigg\{\hspace{-2bp}\bigg(\dfrac{X-\mu}{\sigma}\bigg)^{3}\bigg\}\bigg|,
\end{equation}
which is nothing but the \textit{absolute value of Pearson's moment coefficient
of skewness} (i.e., excluding directionality). In other words, Pearson's moment cofficient of skewness may be interpreted itself as the \textit{difference of a pair of optimal estimators; these are the mean of $X$ (in the MMSE sense), and the maximally risk-averse estimator of $X$, optimally biased towards the tail of the distribution $\mathcal{P}_X$}.
Further, via our topological interpretation of $d(\mathcal{P})$, Pearson's moment coefficient of skewness \textit{expresses, in absolute value, the distance (in a topologically consistent sense) of the distribution of $X$ relative to any non-skewed distribution on the real line}, with the most obvious representative being $\mathcal{N}(0,1)$.

\begin{figure}[t!]
\begin{center}
\begin{tabular}{c}
\hspace{-0.8cm}
{\includegraphics[width=0.53\textwidth,height=0.23 \textwidth, trim={0.1cm 0cm 1.6cm 0cm},clip]{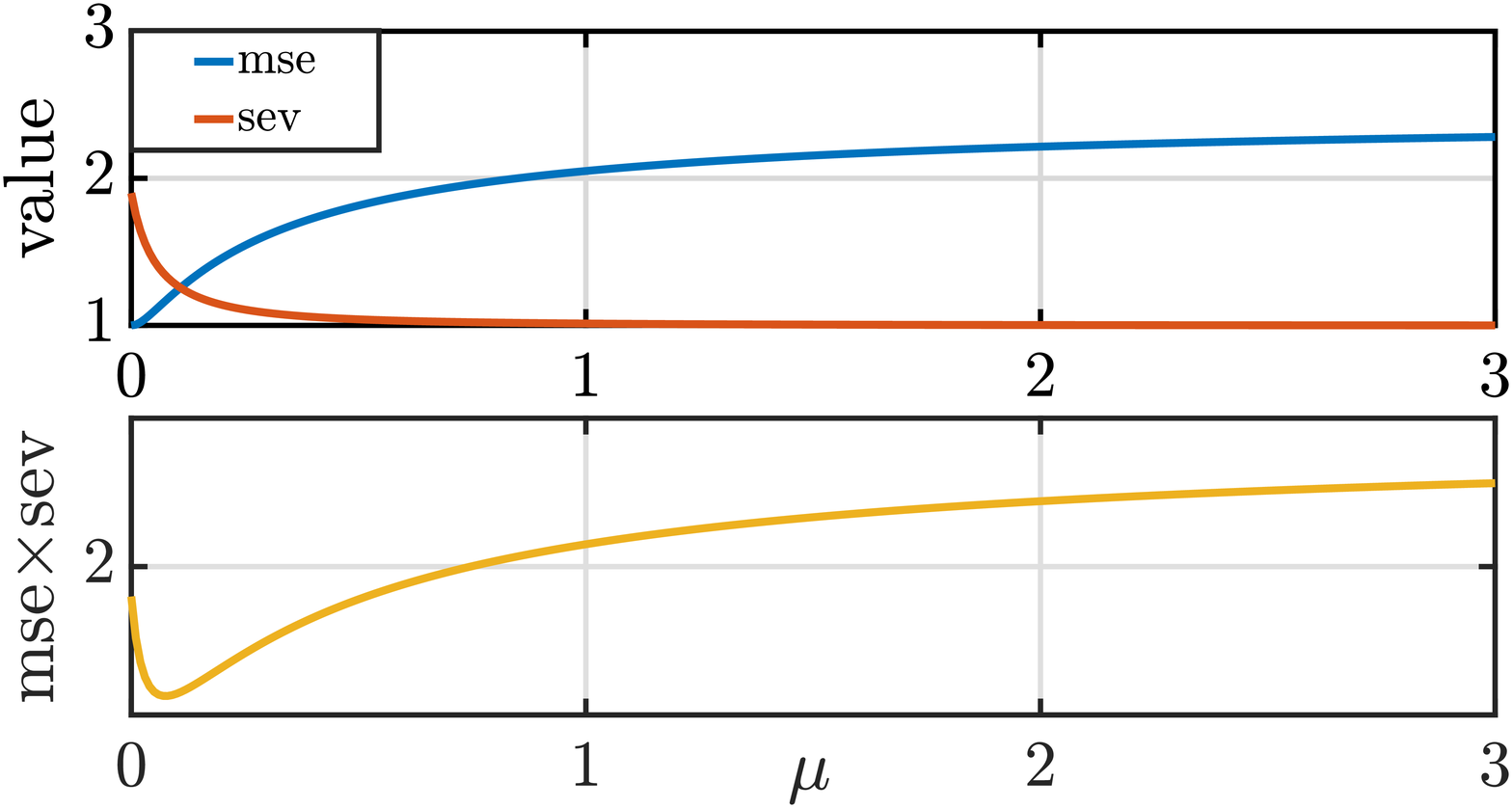}}
\end{tabular}
\end{center}
\vspace{-8pt}
\caption{ Normalized $(\mathrm{mse})/(\mathrm{sev})$ and their corresponding product in case of state-dependent noise.}
\label{fig0}
\end{figure}
\par
Consequently, the risk margin functional $d$ (intuitively a scalar quantity) may be thought as a measure of skewness \textit{magnitude} in multiple dimensions, corresponding to a consistent non-directional generalization of Pearson's moment skewness coefficient, also fully applicable the hidden state model setting, and tacitly exploiting statistical dependencies of both the conditional and marginal measures
${\cal P}_{\boldsymbol{X}|\boldsymbol{Y}}$ and ${\cal P}_{\boldsymbol{Y}}$. In the same fashion, the risk margin (pseudo)metric $d_{\mathbb{S}}$ may be
conventiently thought as a measure of the \textit{relative skewness}
\textit{between (filtered) distributions}.
\par
Of course, skewness directionality, while informative, is a much more complicated concept in multiple dimensions as compared to the case of random variables on the line, where directionality reduces to the sign of a centered third-order moment. Nonetheless, while directionality is naturally not captured by the distance-related functionals $d$ and $d_{\mathbb{S}}$, it is embedded in the random vector $\dx$, and is  \textit{optimally exploited} by the family of optimal risk-aware estimations $\{\hat{\boldsymbol{X}}_{\mu}^{*}\}_{\mu}$, for each observation $\boldsymbol{Y}$.



\section{Numerical Simulations and Discussion}
\label{section7}
Our theoretical claims are now justified through indicative numerical illustrations, along with a relevant discussion. We justify our claims by presenting the following working examples: 
First, we consider the problem of inferring an exponentially distributed hidden state $X$, with $\mathbb{E}\{X\}=2$ while observing $Y=X+v$ \cite{kalogerias2020better}. The random variable $v$ expresses a state-dependent, zero-mean, normally distributed noise, whose (conditional) variance is given by $\mathbb{E}\{v^2|X\}=9X^2$. 
Fig. \ref{fig0} shows $\mathrm{mse}(\xemi)$, $\mathrm{sev}(\xemi)$, as well as their product $\mathrm{mse}(\xemi)\mathrm{sev}(\xemi)$, all with respect to the risk-aversion parameter $\mu$. The former two have been normalized with respect to their corresponding minimum values while the product results after the aforementioned normalization step. From the figure, it is evident that the optimal trade-off (in the sense implied by Theorem \ref{UP}) is attained close to the origin; note, though, that such an optimal $\mu^*$ does \textit{not} correspond to the value of $\mu$ for which (normalized) $\mathrm{mse}$ and $\mathrm{sev}$ curves intersect. 
\par
\begin{figure}[t!]
\begin{center}
\begin{tabular}{c}
\hspace{-0.6cm}
{\includegraphics[width=0.52\textwidth,height=0.23 \textwidth, trim={0.1cm 0cm 1.6cm 0cm},clip]{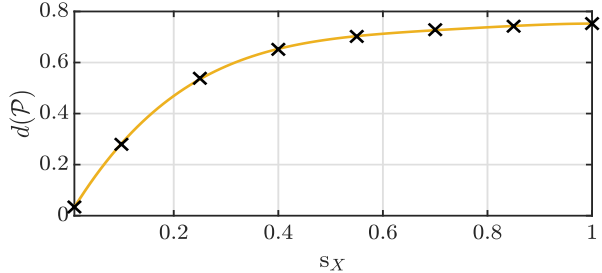}}
\end{tabular}
\end{center}
\vspace{-8pt}
\caption{Measure of skewness as a function of the parameter $\mathrm{s}_{X}$.}
\label{fig5}
\vspace{-4pt}
\end{figure}
Next, we consider the problem of estimating another real-valued hidden state $X$ while observing $Y=X \times W$,
with 
$
    (X,W)\sim \mathrm{Lognormal}(\boldsymbol{0},\boldsymbol{S}),
$
$\boldsymbol{S}=\mathrm{diag}(\mathrm{s}_{X},0.25)$. The variable $\mathrm{s}_{X}>0$ defines a parametric family of probability measures whose skewness increases with $\mathrm{s}_{X}$. 
We would like to examine the impact of our theoretical results by varying the skewness of the aforementioned model. However, we are not aware that by increasing $\mathrm{s}_{X}$ the posterior skewness alters as well. In addition, even if skewness varies with $\mathrm{s}_{X}$, the way it does so is not apparent. For these reasons, we employ our new distance/skewness measure to trial the model with respect to $\mathrm{s}_{X}$. This experiment is shown in Fig. \ref{fig5} where we verify that at least for the examined $\mathrm{s}_{X}$-values, the average posterior skewness increases. 
\par
Fig. \ref{fig3} illustrates how the profiles of $\mathrm{mse}(\xemi)$, $\mathrm{sev}(\xemi)$, and their product $\mathrm{mse}(\xemi)\mathrm{sev}(\xemi)$ scale with $\mathrm{s}_{X}$. As above, we normalize $\mathrm{mse}(\xemi)$ and $\mathrm{mse}(\xemi)\mathrm{sev}(\xemi)$ with respect to their minimum possible value, respectively, and $\mathrm{sev}(\xemi)$ with respect to its maximum one.
First, 
although the average performance deteriorates faster as the skewness increases (e.g., for the most skewed model, depicted in cyan), a $15\%$ deterioration of $\mathrm{mse}$ corresponds to a $20\%$ safety improvement, indicating that, there might be particular models allowing for an even more advantageous exchange. 
\par 
Further,  Fig.~\ref{fig3} shows that,
for the smallest skewness level (blue), almost all risk-aware estimates achieve a near-optimal bound. As the skewness increases, the optimal -with respect to the product- estimators become strongly separated from each other within the class $\{\xemi\}_\mu$. In this one-dimensional example, there is a unique optimal value for $\mu^{\star}$ with respect to the product; however, this might be only an exception to the rule, especially for higher-dimensional models. Note that a graphical representation of the product like the one depicted in Fig.~\ref{fig3} is all that we need to do to at least approximately determine the optimal value for $\mu$ (a single parameter).
\par
Lastly, Fig.~\ref{fig4} presents the course of the upper bound $\mathbb{U}(\mathcal{P})$ with respect to the skewness parameter $\mathrm{s}_{X}$. To clarify its behavior close to zero, we sample $\mathrm{s}_{X}$ additionally at $0.01$ and $0.1$. Expectedly, while $d(\mathcal{P})$ approaches zero, the bound approaches zero as well regardless of the chosen limit $\rho_{\max}$, and the values $\mathrm{mse}(\xeo)$, and $\mathrm{sev}(\xeinf)$.

\begin{figure}[t!]
\begin{adjustwidth}{3cm}{0cm}
\hspace{8.5bp}
\begin{tabular}{c}
{\includegraphics[width=0.53\textwidth,height=0.23 \textwidth, trim={0.1cm 0cm 1.6cm 0cm},clip]{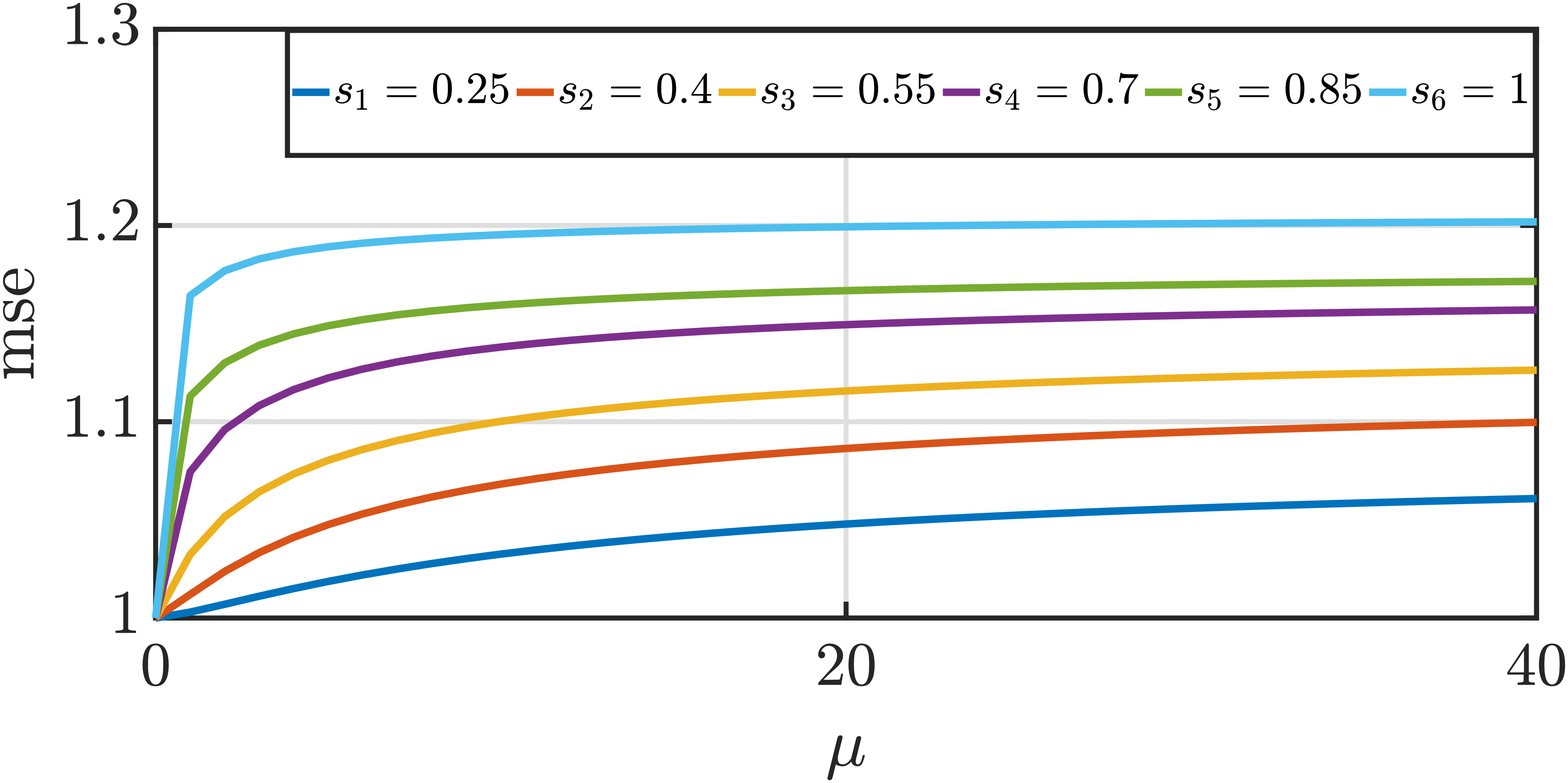}}
\vspace{-23bp}
\\
{\includegraphics[width=0.53\textwidth,height=0.23 \textwidth, trim={0.1cm 0cm 1.6cm 0cm},clip]{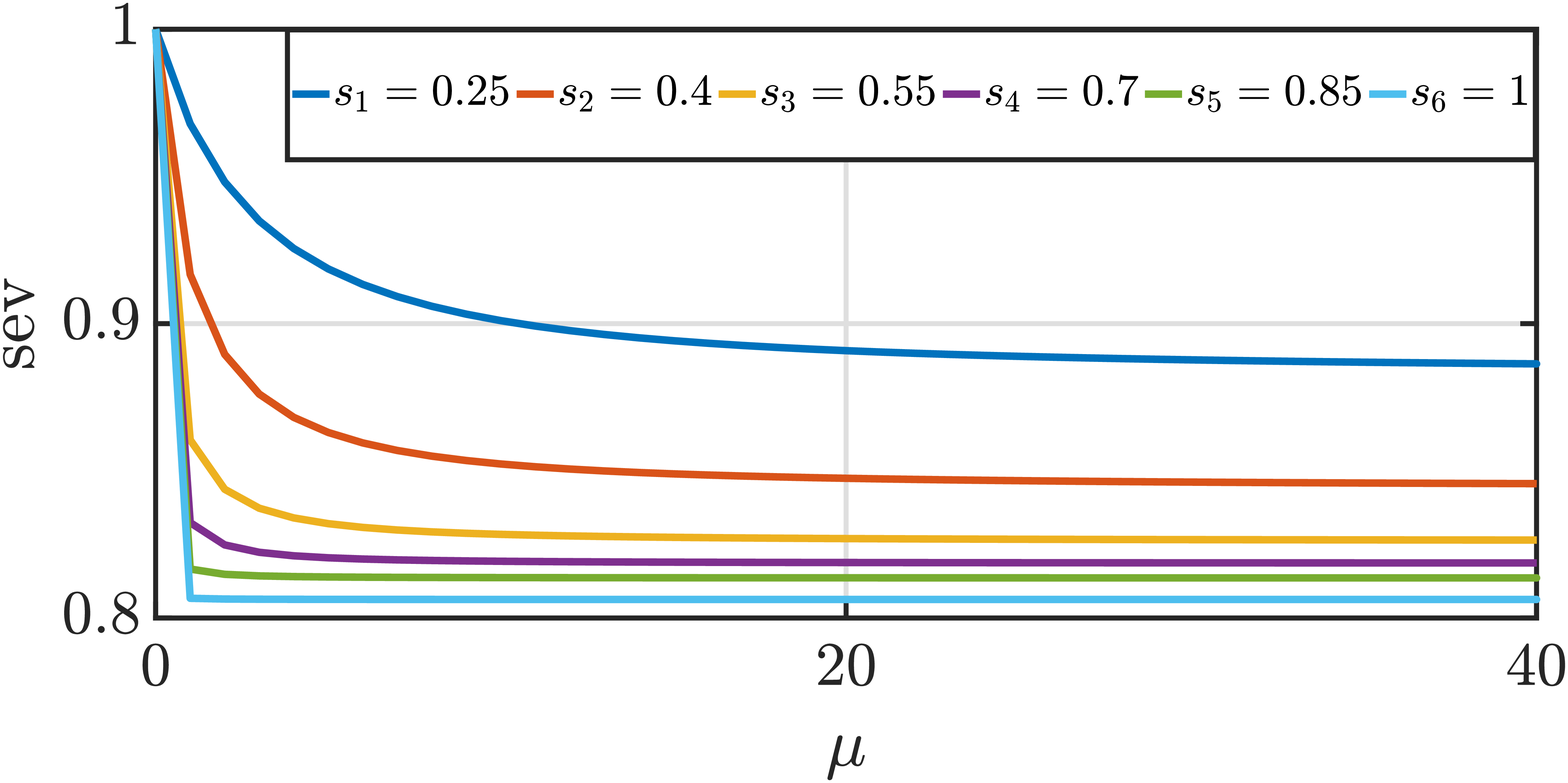}}
\vspace{-23bp}
\\
\hspace{-6.6bp}
{\includegraphics[width=0.53\textwidth,height=0.23 \textwidth, trim={0.1cm 0cm 1.6cm 0cm},clip]{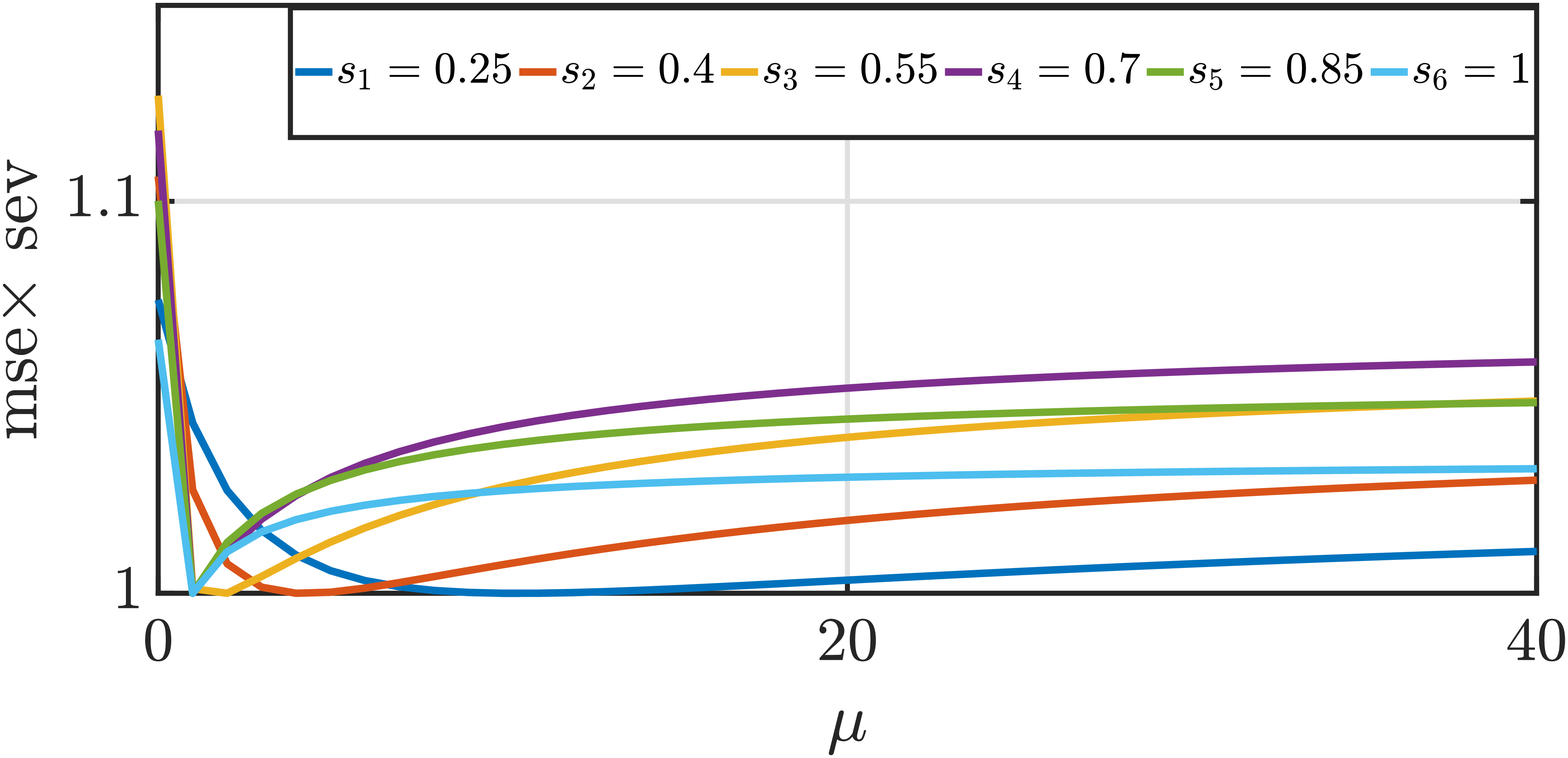}}
\end{tabular}
\end{adjustwidth}
\vspace{-4pt}
\caption{Up (Center): $\mathrm{mse}$ ($\mathrm{sev}$) percent increase (decrease) relative to risk-aversion parameter $\mu$ for different skewness levels. Down: Normalized trade-off relative to risk-aversion parameter $\mu$ for different skewness levels.}
\label{fig3}
\vspace{-4pt}
\end{figure}
\section{Conclusion}
This work quantified the inherent trade-off between $\mathrm{mse}$ and $\mathrm{sev}$ by lower bounding the product between the two over all admissible estimators. Provided a level of performance (resp. risk), the introduced uncertainty relation reveals the minimum risk (resp. performance) tolerance for the problem and assesses how effective any estimator is with respect to the optimal Bayesian trade-off.
Projecting the risk-averse stochastic $\mu$-parameterized curve on the link between the MMSE and the maximally risk-averse estimator, we defined as analyzed the so-called hedgeable risk margin of the model. Its significance stems from the fact that it admits both a rigorous topological and an intuitive statistical interpretations, fitting our risk-aware estimation setting. In particular, the risk margin functional induces a new measures of the skewness of the conditional evidence regarding the state provided the observables. 
Connecting the dots, we showed that the optimal trade-off is order-equivalent to this new measure of skewness, thus fully characterizing our uncertainty principle from a statistical perspective.
\begin{figure}[t!]
\begin{center}
\begin{tabular}{c}
\hspace{-0.6cm}
{\includegraphics[width=0.52\textwidth,height=0.23 \textwidth, trim={0.1cm 0cm 1.6cm 0cm},clip]{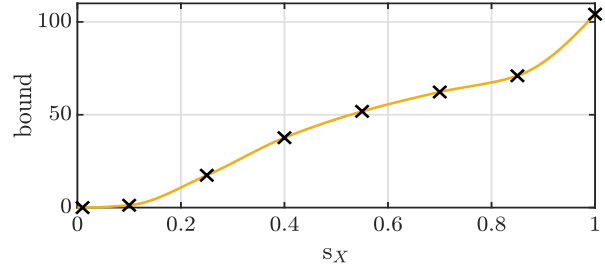}}
\end{tabular}
\end{center}
\vspace{-8pt}
\caption{The bound $\mathbb{U}(\mathcal{P})$ as a function of the parameter $\mathrm{s}_{X}$. In this example $\rho_{\max}=10$. }
\label{fig4}
\vspace{-4pt}
\end{figure}
\section{Appendix}

\subsection{Maximally Risk-Averse MMSE estimator}
First, recall that 
\begin{align}
\boldsymbol{\Sigma}_{\boldsymbol{X} \mid \boldsymbol{Y}} \triangleq \mathbb{E}\left\{(\boldsymbol{X}-\mathbb{E}\{\boldsymbol{X} \mid \boldsymbol{Y}\})(\boldsymbol{X}-\mathbb{E}\{\boldsymbol{X} \mid \boldsymbol{Y}\})^{\top} \mid \boldsymbol{Y}\right\} =\boldsymbol{U}\boldsymbol{\Lambda} \boldsymbol{U}^{\top},
\end{align}
which implies 
\begin{align}
 \boldsymbol{U}^{\top} \boldsymbol{\Sigma}_{\boldsymbol{X} \mid \boldsymbol{Y}}\boldsymbol{U}&=\mathbb{E}\left\{(\boldsymbol{U}^{\top}\boldsymbol{X}-\boldsymbol{U}^{\top}\mathbb{E}\{\boldsymbol{X} \mid \boldsymbol{Y}\})(\boldsymbol{U}^{\top}\boldsymbol{X}-\boldsymbol{U}^{\top}\mathbb{E}\{\boldsymbol{X} \mid \boldsymbol{Y}\})^{\top} \mid \boldsymbol{Y}\right\}\nonumber\\
 &=\boldsymbol{\Lambda}, \label{kyriakos }
\end{align}
where we assume that $\boldsymbol{\Lambda}$ conserves only $r$ non-zero eigenvalues, i.e,

\begin{align}
\boldsymbol{\Lambda}=\mathrm{diag}\big( \big\{ \sigma_{i}(\boldsymbol{Y})\big\}_{i\in\mathbb{N}^{+}_r},\boldsymbol{0}\big).\label{mitsotakis}
\end{align}
From \eqref{kyriakos } and \eqref{mitsotakis} we may infer that 
\begin{align}
    [\boldsymbol{U}^{\top}(\boldsymbol{X}-\mathbb{E}\{\boldsymbol{X}|\boldsymbol{Y}\})]_{i}^2=0~~,i=r+1,...,n,
\end{align}
or equivalently that 
\begin{align}
    [\boldsymbol{U}^{\top}\boldsymbol{X}]_{i}=[\boldsymbol{U}^{\top}\mathbb{E}\{\boldsymbol{X}|\boldsymbol{Y}\}]_{i}~~, i=r+1,...,n. \label{evdomindadyo}
\end{align}
Recalling \eqref{adonis}, \eqref{eq:RA_MMSE} we have 
\begin{align}
&\xemi\nonumber\\
&=(\boldsymbol{I}+2\mu\boldsymbol{\Sigma}_{\boldsymbol{X} \mid \boldsymbol{Y}})^{-1}[ \mathbb{E}\{\boldsymbol{X} \mid \boldsymbol{Y}\}+\mu\big(\mathbb{E}\{\|\boldsymbol{X}\|_{2}^{2} \boldsymbol{X} \mid \boldsymbol{Y}\}-\mathbb{E}\{\|\boldsymbol{X}\|_{2}^{2} \mid \boldsymbol{Y}\} \mathbb{E}\{\boldsymbol{X} \mid \boldsymbol{Y}\}\big)]\nonumber\\[5pt]
&=\boldsymbol{U}\Big(\mathbb{I}+2\mu\boldsymbol{\Lambda} \Big)^{-1}
\left\{ 
\left[\begin{array}{c}\hspace{-0.4cm}\left[
\mathbb{E}\{[\boldsymbol{U}^{\top}\boldsymbol{X}]_{i}|\boldsymbol{Y} \}\right]^{1}_{r}\\[\medskipamount]
\left[
\mathbb{E}\{[\boldsymbol{U}^{\top}\boldsymbol{X}]_{i}|\boldsymbol{Y} \}\right]^{r+1}_{n}
\end{array}\hspace{-0.1cm}\right]+\mu
\left[\begin{array}{c}\hspace{-0.4cm}\left[
\mathbb{E}\{||\boldsymbol{X}||^2_2[\boldsymbol{U}^{\top}\boldsymbol{X}]_{i}|\boldsymbol{Y} \}-\mathbb{E}\{||\boldsymbol{X}||^2_2|\boldsymbol{Y}\}[\boldsymbol{U}^{\top}\mathbb{E}\{\boldsymbol{X}|\boldsymbol{Y}\}]_{i} \}
\right]^{1}_{r}\\[\medskipamount]
\hspace{-0.1cm}\left[\mathbb{E}\{||\boldsymbol{X}||^2_2|\boldsymbol{Y}\}[\boldsymbol{U}^{\top}\boldsymbol{X}]_{i}-\mathbb{E}\{||\boldsymbol{X}||^2_2|\boldsymbol{Y}\}[\boldsymbol{U}^{\top}\mathbb{E}\{\boldsymbol{X}|\boldsymbol{Y}\}]_{i}
\right]^{r+1}_{n}
\end{array}\hspace{-0.1cm}\right]
\right\} \nonumber \\[5pt]
&=\boldsymbol{U}\Big(\mathbb{I}+2\mu\boldsymbol{\Lambda} \Big)^{-1}\left\{ \left[\begin{array}{c}\hspace{-0.4cm}\left[
\mathbb{E}\{[\boldsymbol{U}^{\top}\boldsymbol{X}]_{i}|\boldsymbol{Y} \}\right]^{1}_{r}\\[\medskipamount]
\left[
\mathbb{E}\{[\boldsymbol{U}^{\top}\boldsymbol{X}]_{i}|\boldsymbol{Y} \}\right]^{r+1}_{n}
\end{array}\hspace{-0.1cm}\right]+\mu
\left[\begin{array}{c}\hspace{-0.1cm}\left[
\mathbb{E}\{||\boldsymbol{X}||^2_2[\boldsymbol{U}^{\top}\boldsymbol{X}]_{i}|\boldsymbol{Y} \}-\mathbb{E}\{||\boldsymbol{X}||^2_2|\boldsymbol{Y}\}[\boldsymbol{U}^{\top}\mathbb{E}\{\boldsymbol{X}|\boldsymbol{Y}\}]_{i} \}
\right]^{1}_{r}\\[\medskipamount]
\hspace{-0.1cm}\left[[0]_{i}
\right]^{r+1}_{n}
\end{array}\hspace{-0.1cm}\right]
\right\} \label{skatastaarnia}
\end{align}
where the latter holds since \eqref{evdomindadyo}, indicates that 
\begin{align}
  \mathbb{E}\{\|\boldsymbol{X}\|_{2}^{2}[\boldsymbol{U}^{\top} \boldsymbol{X}]_{i} \mid \boldsymbol{Y}\}-\mathbb{E}\{\|\boldsymbol{X}\|_{2}^{2} \mid \boldsymbol{Y}\}[\boldsymbol{U}^{\top} \mathbb{E}\{\boldsymbol{X} \mid \boldsymbol{Y}\}]_{i}&\nonumber\\[5pt]
   & \hspace{-5cm}=\mathbb{E}\{\|\boldsymbol{X}\|_{2}^{2}\mid \boldsymbol{Y}\}[\boldsymbol{U}^{\top} \boldsymbol{X}]_{i}-\mathbb{E}\{\|\boldsymbol{X}\|_{2}^{2} \mid \boldsymbol{Y}\}[\boldsymbol{U}^{\top} \mathbb{E}\{\boldsymbol{X} \mid \boldsymbol{Y}\}]_{i}\nonumber\\[5pt]
   & \hspace{-5cm}=\mathbb{E}\{\|\boldsymbol{X}\|_{2}^{2}\mid \boldsymbol{Y}\}[\boldsymbol{U}^{\top} \mathbb{E}\{\boldsymbol{X} \mid \boldsymbol{Y}\}]_{i}-\mathbb{E}\{\|\boldsymbol{X}\|_{2}^{2} \mid \boldsymbol{Y}\}[\boldsymbol{U}^{\top} \mathbb{E}\{\boldsymbol{X} \mid \boldsymbol{Y}\}]_{i}\nonumber\\[5pt]
   &\hspace{-5cm}=0~,~~i=r+1,...,n.
\end{align}
Further, notice that from  \eqref{mitsotakis} and \eqref{skatastaarnia} we obtain 
\begin{align}
\lim_{\mu\rightarrow\infty}\boldsymbol{U}\Big(\mathbb{I}+2\mu\boldsymbol{\Lambda} \Big)^{-1}
\left[\begin{array}{c}\hspace{-0.4cm}\left[
\mathbb{E}\{[\boldsymbol{U}^{\top}\boldsymbol{X}]_{i}|\boldsymbol{Y} \}\right]^{1}_{r}\\[\medskipamount]
\left[
\mathbb{E}\{[\boldsymbol{U}^{\top}\boldsymbol{X}]_{i}|\boldsymbol{Y} \}\right]^{r+1}_{n}
\end{array}\hspace{-0.1cm}\right]
=\boldsymbol{U}\left[\begin{array}{c}\mathbf{0}_{r} \\ {\left[\boldsymbol{U}^{\top} \mathbb{E}\{\boldsymbol{X} \mid \boldsymbol{Y}\}\right]_{n}^{r+1}}\end{array}\right]
\end{align}
and 
\begin{align}
&\lim_{\mu\rightarrow\infty}\boldsymbol{U}\Big(\mathbb{I}+2\mu\boldsymbol{\Lambda} \Big)^{-1}\mu \left[\begin{array}{c}\hspace{-0.1cm}\left[
\mathbb{E}\{||\boldsymbol{X}||^2_2[\boldsymbol{U}^{\top}\boldsymbol{X}]_{i}|\boldsymbol{Y} \}-\mathbb{E}\{||\boldsymbol{X}||^2_2|\boldsymbol{Y}\}[\boldsymbol{U}^{\top}\mathbb{E}\{\boldsymbol{X}|\boldsymbol{Y}\}]_{i} \}
\right]^{1}_{r}\\[\medskipamount]
\hspace{-0.1cm}\left[[0]_{i}
\right]^{r+1}_{n}
\end{array}\hspace{-0.1cm}\right]
\nonumber\\[5pt]
&=\left[\begin{array}{c}
\frac{1}{2\sigma_{1}(\boldsymbol{Y})}\Big(\mathbb{E}\{||\boldsymbol{X}||^2_2[\boldsymbol{U}^{\top}\boldsymbol{X}]_{1}|\boldsymbol{Y} \}-\mathbb{E}\{||\boldsymbol{X}||^2_2|\boldsymbol{Y}\}[\boldsymbol{U}^{\top}\mathbb{E}\{\boldsymbol{X}|\boldsymbol{Y}\}]_{1} \}\Big)\\
\vdots\\
\frac{1}{2\sigma_{r}(\boldsymbol{Y})}\Big(\mathbb{E}\{||\boldsymbol{X}||^2_2[\boldsymbol{U}^{\top}\boldsymbol{X}]_{r}|\boldsymbol{Y} \}-\mathbb{E}\{||\boldsymbol{X}||^2_2|\boldsymbol{Y}\}[\boldsymbol{U}^{\top}\mathbb{E}\{\boldsymbol{X}|\boldsymbol{Y}\}]_{r} \}\Big)\\
\vdots\\
{0}
\end{array}\right]\nonumber\\
&=\frac{1}{2}\boldsymbol{\Sigma}_{\boldsymbol{X}|\boldsymbol{Y}}^{\dagger}\left(\mathbb{E}\left\{\|\boldsymbol{X}\|_{2}^{2} \boldsymbol{X} \mid \boldsymbol{Y}\right\}-\mathbb{E}\left\{\|\boldsymbol{X}\|_{2}^{2} \mid \boldsymbol{Y}\right\} \mathbb{E}\{\boldsymbol{X} \mid \boldsymbol{Y}\}\right)\nonumber\\[5pt]
&=\frac{1}{2}\boldsymbol{\Sigma}_{\boldsymbol{X}|\boldsymbol{Y}}^{\dagger}\boldsymbol{R}(\boldsymbol{Y})
\end{align}
which yields 
\begin{equation}
\hat{\boldsymbol{X}}_{\infty}^{*}(\boldsymbol{Y})=\dfrac{1}{2}\boldsymbol{\Sigma}_{\boldsymbol{X}|\boldsymbol{Y}}^{\dagger}\boldsymbol{R}(\boldsymbol{Y})+
\boldsymbol{U}\begin{bmatrix}{\bf 0}_{r}\\
\big[\boldsymbol{U}^{\top}\mathbb{E}\big\{\hspace{-1pt}\boldsymbol{X}|\boldsymbol{Y}\big\}\big]_{n}^{r+1}
\end{bmatrix},
\end{equation}
and concludes the proof \hfill $\square$

\subsection{Proof of lemma 1}

Consider $\hat{\boldsymbol{X}}_{\mu}^{*}$ and $\hat{\boldsymbol{X}}^{*}_{\mu'}$ being minimizers of \eqref{eq:Lagrangian} with $\mu \neq \mu'$. Then we may write: 
\begin{align}
&\mathrm{mse}(\hat{\boldsymbol{X}}^{*}_{\mu'})+\mu'\,\mathrm{sev}(\hat{\boldsymbol{X}}^{*}_{\mu'}) \leq \mathrm{mse}(\hat{\boldsymbol{X}}^{*}_{\mu})+\mu'\,\mathrm{sev}(\hat{\boldsymbol{X}}^{*}_{\mu})~,\label{min1}
\end{align}

and
\begin{align}
&\mathrm{mse}(\hat{\boldsymbol{X}}^{*}_{\mu})+\mu\,\mathrm{sev}(\hat{\boldsymbol{X}}^{*}_{\mu}) \leq \mathrm{mse}(\hat{\boldsymbol{X}}^{*}_{\mu'})+\mu\,\mathrm{sev}(\hat{\boldsymbol{X}}^{*}_{\mu'})~.\label{min2}
\end{align}
By adding \eqref{min1} and \eqref{min2} we obtain:
\begin{align}
(\mu-\mu')\mathrm{sev}(\hat{\boldsymbol{X}}^{*}_{\mu}) \leq (\mu-\mu')\mathrm{sev}(\hat{\boldsymbol{X}}^{*}_{\mu'}) ~,
\end{align}
which shows that $\mathrm{sev}(\hat{\boldsymbol{X}}_{\mu})$ decreases w.r.t.\ $\mu$. Furthermore, from either \eqref{min1} or \eqref{min2} we obtain that  $\mathrm{mse}(\hat{\boldsymbol{X}}_{\mu})$ is increasing. For example, by assuming $\mu - \mu'>0$, \eqref{min1} reads: 
\begin{equation}\label{eq13}
\begin{aligned}
&\mathrm{mse}(\hat{\boldsymbol{X}}^{*}_{\mu'})- \mathrm{mse}(\hat{\boldsymbol{X}}^{*}_{\mu})\leq \mu'(\mathrm{sev}(\hat{\boldsymbol{X}}^{*}_{\mu})-\mathrm{sev}(\hat{\boldsymbol{X}}^{*}_{\mu'}))\leq0
\end{aligned}
\end{equation}
\hfill $\square$
\subsection{Proof of lemma 2}
To begin with let us recall that 
\begin{align}
\frac{d\xemi(\boldsymbol{Y})}{d\mu}=2\boldsymbol{U}\boldsymbol{\Lambda}(\boldsymbol{Y})^2\boldsymbol{D}_{\cc}\boldsymbol{U}^{\top}\dx. \label{dif}
\end{align}
By integrating \eqref{dif} in $(\mu,\mu')$ we obtain 
\begin{align}
\xemi-\xemit=(\mu-\mu')\boldsymbol{U}\boldsymbol{H}(\mu,\mu')\boldsymbol{U}^{\top}\dx~,\label{difference}
\end{align}
where 
\begin{align}
 &\hspace{-4bp}\boldsymbol{H}(\mu,\mu')  
=\mathrm{diag}\Bigg( \bigg\{\frac{2\siy}{\big(1+2\mu\siy\big)\big(1+2\mu'\siy\big)}\bigg\}_{i\in{\mathbb{N}^+_r}},\boldsymbol{0}\Bigg). \label{matrixh}
\end{align}
To show Lipschitz for $\mathrm{mse}(\xemi)$ in $[0,+\infty)$, consider the absolute difference 
\begin{align}
|\mathrm{mse}(\xemi)-\mathrm{mse}(\xemit)|=\big|\mathbb{E}\big\{\|\xemi\|_{2}^{2} -2{\hat{\boldsymbol{X}}_{0}^{* \top}}\xemi- \|\xemit\|_{2}^{2} +2{\hat{\boldsymbol{X}}_{0}^{* \top}}\xemit \big\}\big|,
\end{align}
and subsequently add and subtract ${\hat{\boldsymbol{X}}_{{\mu}^{ \prime}}^{* \top}}\xemi$ within the expectation to obtain
\begin{align}
|\mathrm{mse}(\xemi)-\mathrm{mse}(\xemit)|=\big|\mathbb{E}\big\{(\xemi-\xeo+\xemit-\xeo)^{\top}(\xemi-\xemit)\big\}\big|.
\end{align}
Thus, by employing \eqref{difference} and \eqref{matrixh} we may write 
\begin{align}
|\mathrm{mse}(\xemi)-\mathrm{mse}(\xemit)|&=\big|\mathbb{E}\big\{(\xemi-\xeo+\xemit-\xeo)^{\top}(\xemi-\xemit)\big\}\big|\nonumber\\[5pt]
&=4|\mu-\mu'|~\Big|\mathbb{E}\Big\{\dx^{\top}\boldsymbol{U}\Big(\frac{\mu}{2}\boldsymbol{H}(\mu,0)+\frac{\mu'}{2}\boldsymbol{H}(\mu',0) \Big) \boldsymbol{H}(\mu,\mu')\boldsymbol{U}^{\top}\dx \Big\}\Big|.\label{treno}
\end{align}
For $\kappa\in\{\mu, \mu'\}$ we have 
\begin{align}
\big[\boldsymbol{U}^{\top}\dx\big]^{\top}\frac{\kappa}{2}\boldsymbol{H}(\kappa,0)\boldsymbol{H}(\mu,\mu')\big[\boldsymbol{U}^{\top}\dx\big]
&\leq \max\Bigg\{\Bigg\{[\frac{\kappa}{2}\boldsymbol{H}(\kappa,0)\boldsymbol{H}(\mu,\mu')]_{i,i}\Bigg\}_{1 \leq i \leq r},0 \Bigg\}||\dx||^2_2 \nonumber\\[5pt]
&=\max\Bigg\{[\frac{\kappa}{2}\boldsymbol{H}(\kappa,0)\boldsymbol{H}(\mu,\mu')]_{i,i}\Bigg\}_{1 \leq i \leq r}||\dx||^2_2~.\label{paparia}
\end{align}
Further, since 
\begin{align}
 [\boldsymbol{H}(\kappa,0)]_{i,i}&=\frac{\kappa\siy}{1+2\kappa\siy}\nonumber\\[5pt]
 &<\frac{1}{2}~,
\end{align} 
and
\begin{align}
 [\boldsymbol{H}(\mu,\mu')]_{i,i}&=\frac{2\siy}{(1+2\mu\siy)(1+2\mu'\siy)}\nonumber\\[5pt]
 &<2\siy~,
\end{align}
\eqref{treno}, and \eqref{paparia} yield
\begin{align}
\big|\mathrm{mse}(\hat{\boldsymbol{X}}^{*}_\mu)-\mathrm{mse}(\hat{\boldsymbol{X}}^{*}_{\mu'})\big|& \leq4|\mu-\mu'|~\mathbb{E}\Big\{\max_{1 \leq i \leq r}\siy\big\|\dx \big\|^2_2 \Big\}\nonumber\\[4pt]
&=4|\mu-\mu'|~\mathbb{E}\big\{\smaxy\big\|\dx\big\|^2_2\big\}
\end{align}
Continuity of $\mathrm{mse}(\xemi)$ to $+\infty$ follows after applying $\mu=+\infty$ in \eqref{dif}, and from the fact that 
\begin{align}
\big|\mathrm{mse}(\hat{\boldsymbol{X}}^{*}_\mu)-\mathrm{mse}(\hat{\boldsymbol{X}}^{*}_{\infty})\big|&=|\mathbb{E}\{ (\xemi+\xeinf-2\xeo)^{\top}(\xemi-\xeinf)\}|\nonumber\\[5pt]
&=\big|\mathbb{E}\{ (\xemi-\xeo-(\xeo-\xeinf)) (\xemi-\xeinf ) \}\big|\nonumber\\[5pt]
&=\Bigg|\mathbb{E}\{\dx^{\top}\boldsymbol{U}
\mathrm{diag}\Bigg(\Big\{ \frac{1}{\big(1+2\mu\siy\big)^2}\Big\}_{i\in{\mathbb{N}^+_r}},
\boldsymbol{0}\Bigg)
\boldsymbol{U}^{\top}\dx \}\Bigg| \nonumber\\[5pt]
&\leq\mathbb{E}\Bigg\{ \max_{1 \leq i \leq r}\Bigg\{\frac{1}{(1+2\mu\siy)^2} \Bigg\} ||\dx||^2_2 \Bigg\} \nonumber\\[5pt]
&\leq \frac{1}{4\mu^2}\mathbb{E}\Bigg\{ \frac{||\dx||^2_2}{{\sminy}^2} \Bigg\}
\end{align}
In a similar fashion, we may show Lipschitz for $\mathrm{sev}(\xemi)$ in $[0,+\infty)$ by considering the absolute difference 

\begin{align}
|\mathrm{sev}(\xemi)-\mathrm{sev}(\xemit)|
&=
\mathbb{E}\big\{ \|\xemi\|_{\cc}^{2} -2{\hat{\boldsymbol{X}}_{\infty}^{* \top}}{\cc}{\xemi}-  \|\xemit\|_{\cc}^{2}+2{\hat{\boldsymbol{X}}_{\infty}^{* \top}}{\cc}{\xemit}\big\},
\end{align}
which after adding and subtracting ${\hat{\boldsymbol{X}}_{{\mu}^{ \prime}}^{* \top}}\cc\xemi$ within the expectation reads 
\begin{align}
&\hspace{-0.1cm}|\mathrm{sev}(\hat{\boldsymbol{X}}_\mu)-\mathrm{sev}(\hat{\boldsymbol{X}}_{\mu'})|=\big|\mathbb{E}\{ (\xemi-\xeinf+\xemit-\xeinf)^{\top}\cc(\xemi-\xemit)\}\big|. \label{mlkia}
\end{align}
After applying \eqref{difference} for the appearing differences and subsequently declaring
\begin{align}
\boldsymbol{d}(\mu)=
\mathrm{diag}\Bigg(
\Bigg\{
\frac{1}{1+2\mu\siy}
\Bigg\}_{i\in{\mathbb{N}^+_r}},\boldsymbol{0}
\Bigg),
\end{align}
and 
\begin{align}
\boldsymbol{D}(\mu,\mu')=
\mathrm{diag}\Bigg(\Bigg\{\frac{2\siy^2}{(1+2\mu\siy)(1+2\mu'\siy)}\Bigg\}_{i\in{\mathbb{N}^+_r}},\boldsymbol{0}\Bigg),
\end{align}
equation \eqref{mlkia} yields
\begin{align}
|\mathrm{sev}(\hat{\boldsymbol{X}}_\mu)-\mathrm{sev}(\hat{\boldsymbol{X}}_{\mu'})|=|\mu-\mu^{'}|~\Big|\mathbb{E}\big\{\big[ \boldsymbol{U}^{\top}\dx\big]^{\top} \big(\boldsymbol{d}(\mu)  +\boldsymbol{d}(\mu') \big)
\boldsymbol{D}(\mu,\mu')\big[ \boldsymbol{U}^{\top}\dx\big]\big\}\Big|.\label{enatsigaraki}
\end{align}
For $\kappa\in\{\mu, \mu'\}$ we may write  
\begin{align}
[\boldsymbol{U}^{\top}\dx]^{\top}\boldsymbol{d}(\kappa)\boldsymbol{D}(\mu,\mu')[\boldsymbol{U}^{\top}\dx]&\leq \max\Bigg\{\Bigg\{[\boldsymbol{d}(\kappa)\boldsymbol{D}(\mu,\mu')]_{i,i}\Bigg\}_{1 \leq i \leq r},0 \Bigg\}||\dx||^2_2 \nonumber\\
&=\max\Bigg\{[\boldsymbol{d}(\kappa)\boldsymbol{D}(\mu,\mu')]_{i,i}\Bigg\}_{1 \leq i \leq r}||\dx||^2_2\nonumber\\
&\leq\max_{i \leq i \leq r}\{\siy^2\}||\dx||^2_2~, \label{paparia2}
\end{align}
for each term inside the expectation. Therefore \eqref{enatsigaraki}, and \eqref{paparia2} yield 
\begin{align}
|\mathrm{sev}(\hat{\boldsymbol{X}}_\mu)-\mathrm{sev}(\hat{\boldsymbol{X}}_{\mu'})|\leq4|\mu-\mu'|~\mathbb{E}\Big\{\smaxy^{2}\big\|\dx\big\|^2_2\Big\}
\end{align}
Lastly, continuity of $\mathrm{sev}(\xemi)$ to $+\infty$ results from the fact that 
\begin{align}
|\mathrm{sev}(\hat{\boldsymbol{X}}^{*}_\mu)-\mathrm{sev}(\hat{\boldsymbol{X}}^{*}_{\infty})|&=\big|\mathbb{E}\{ (\xemi-\xeinf)^{\top}\cc(\xemi-\xeinf)\}\big|\nonumber\\[5pt]
& \leq\mathbb{E}\Bigg\{\Bigg|[\boldsymbol{U}^{\top}\dx]^{\top}    
\mathrm{diag}\Bigg(\Bigg\{ \frac{\siy}{\big(1+2\mu\siy\big)^2}\Bigg\}_{i\in{\mathbb{N}^+_r}},\boldsymbol{0}\Bigg)
[\boldsymbol{U}^{\top}\dx]\Bigg|\Bigg\}\nonumber\\[5pt]
&\leq\mathbb{E}\Bigg\{ \max_{1 \leq i \leq r}\Bigg\{\frac{\siy}{(1+2\mu\siy)^2} \Bigg\} ||\dx||^2_2 \Bigg\}\nonumber\\[5pt]
&\leq\mathbb{E}\Bigg\{ \max_{1 \leq i \leq r}\Bigg\{\frac{\siy}{(2\mu\siy)^2} \Bigg\} ||\dx||^2_2 \Bigg\}\nonumber\\[5pt]
&\leq  \frac{1}{4\mu^2}\mathbb{E}\Bigg\{\frac{||\dx||^2_2}{\sminy}\Bigg\}
\end{align}
\hfill $\square$

\bibliographystyle{IEEEbib-abbrev}
\bibliography{refs.bib}
\end{document}